\documentclass[11pt]{article}

\usepackage{graphicx}
\usepackage[utf8]{inputenc}
\usepackage[dvipsnames]{xcolor}
\usepackage{nicefrac}
\usepackage{amsthm}
\usepackage{comment}
\usepackage{subfigure}
\usepackage{mathtools}
\usepackage[style=apa,bibstyle=apa,doi=false,url=false,isbn=false,eprint=false,sortcites=false]{biblatex}
\usepackage[colorlinks=true, urlcolor=blue, linkcolor=blue, citecolor=red]{hyperref}
\usepackage[noabbrev,nameinlink,capitalise]{cleveref}
\usepackage{multirow,array}
\usepackage[margin=1in]{geometry}
\usepackage{setspace}
\usepackage{url}
\usepackage{authblk}
\usepackage{paralist}
\usepackage{bm}
\usepackage{bbm}
\usepackage{amsmath,amssymb}
\usepackage{tikz}
\usepackage{dirtytalk}
\newtheorem{theorem}{Theorem}[section]
\newtheorem{proposition}{Proposition}[section]

\newtheorem{lemma}{Lemma}[section]

\newtheorem{cor}{Corollary}[section]
\DeclareMathOperator*{\argmax}{arg\,max} 	
\allowdisplaybreaks

\theoremstyle{definition}
\newtheorem{definition}{Definition}[section]
\newtheorem{example}{Example}[section]

\title{Screening Under Competition}

\author[1]{Yu-Ting Ho\thanks{\href{mailto:hoyuting@berkeley.edu}{hoyuting@berkeley.edu}}}
\affil[1]{Department of Economics, University of California, Berkeley}

\date{\today}

\begin{document}

    \maketitle
    \begin{abstract}
       We study competition among multiple firms that offer differentiated varieties of the same good to a unit-demand agent. The agent has heterogeneous valuations for goods from different firms. Firms do not observe the agent’s exact valuations, but they know their distribution. Firms simultaneously post menus of contracts, after which the agent chooses a firm and one of its contracts to maximize her utility. This defines a game in which firms aim to  maximize expected revenue. We introduce a sufficient condition, density-regularity, under which each firm’s best response to any arbitrary menu profile posted by its opponents is equivalent to posting a menu that contains only a posted-price contract. Our result is not a direct extension of the canonical Myersonian model with a single seller. The standard argument in the literature breaks down once heterogeneous preferences and competition are introduced. We therefore adopt an optimal-control approach, in which the density-regularity condition is essential for establishing the optimality of posted prices. When this condition fails, posted prices may fail to be a best response. 
    \end{abstract}

     \newpage
    
    \section{Introduction}


Screening is a classic topic in microeconomic theory and market design. The canonical one-dimensional monopoly screening problem is well understood, but moving beyond this benchmark introduces substantial difficulties. First, allowing the agent to have a multidimensional type leads to one of the most challenging domains in mechanism design: Optimal mechanisms are often analytically intractable and exhibit complex features that are difficult to implement in practice \parencite{armstrong1996multiproduct, rochet1998ironing, manelli2007multidimensional}. Second, when multiple firms compete, each firm may condition its strategy on the broader contractual environment, giving rise to a complex communication game \parencite{mcafee1993mechanism, epstein1999revelation, peters2001common}. Establishing tractable equilibria in competing mechanism environments typically requires restrictive assumptions. While the literature struggles with these two aspects separately, many real-world markets involve both features at the same time.

Consider a unit-demand agent who chooses among differentiated varieties of the same good offered by competing firms. Since the agent may have different valuations for different varieties, her private information is naturally multidimensional. Firms compete by posting pricing schemes. This setting captures many familiar markets, such as choosing among airlines for the same route, purchasing a cell phone from competing brands, or watching the same movie at theaters operated by different chains. We are interested in markets in which heterogeneous preferences and competition arise simultaneously. Under our formulation, the interaction between these two features becomes tractable. We introduce the model in the following.

The agent's preference is described by a valuation vector \(\bm{\theta}=(\theta_1,\dots,\theta_n)\in\mathbb{R}^n\), where \(n\ge2\) is the number of firms and \(\theta_i\) denotes the agent's valuation for the good offered by firm \(i\). Each \(\theta_i\) is independently drawn from a cumulative distribution function \(F_i\), with strictly positive density \(f_i\) on the support \(\Theta_i=[0,\bar{\theta}_i]\subset\mathbb{R}_{\ge0}\). Firms know the joint distribution of \(\bm{\theta}\), but they do not observe its realization. Before the agent's type is realized, firms compete by offering pricing schemes with the aim of contracting with the agent. 

A \textbf{contract} $(q_i, t_i) \in [0,1] \times \mathbb{R}_{\ge 0}$ offered by firm $i$ consists of an allocation $q_i$, representing the probability or quantity of the good the agent receives, and a transfer $t_i$, representing the payment to firm $i$. Given a realized valuation $\theta_i$, the agent's utility from a contract $(q_i, t_i)$ is $q_i \theta_i - t_i$. A contract is a \textbf{posted-price contract} if it takes the form $(1, p_i)$ for some $p_i \in \mathbb{R}_{\ge 0}$. The contract $(0,0)$ is the \textbf{walk-away contract}, involving no allocation or payment, and is always available to the agent. A \textbf{menu} $M_i \subseteq [0,1] \times \mathbb{R}_{\ge 0}$ is a set of contracts available to the agent. A \textbf{posted-price menu} consists of only the walk-away contract and a posted-price contract: $M_i = \{(0,0), (1, p_i)\}$.

Firms compete by simultaneously posting menus, forming a menu profile $\bm{M} = (M_1, \dots, M_n)$. After observing the menu profile, the agent chooses the contract that yields the highest utility given her valuation vector $\bm{\theta}$. The firm offering the chosen contract then collects the corresponding transfer as revenue. We assume that the agent can contract with only one firm. This assumption naturally captures markets for indivisible goods, where the agent cannot split her unit demand across firms. The strategy space for each firm is the set of menus containing the walk-away contract. A menu profile $\bm{M}$ constitutes an equilibrium if no firm can achieve a strictly higher expected revenue by unilaterally deviating to an alternative menu. We refer to this setting as the \textbf{menu-posting game}. 

Adopting the menu-posting framework is not merely a matter of tractability. It also provides sufficient generality to characterize the equilibrium outcomes of a broader communication game. In this broader game, after the agent's type is realized, she sends a possibly distinct message to each firm. Each firm then proposes an outcome based on the message it receives. Thus, the agent’s strategy maps her type into a profile of messages, while each firm’s strategy maps the message it receives into an outcome. Characterizing equilibria in the broader communication game is difficult. We therefore rely on the fundamental result of \textcite{peters2001common} and \textcite{martimort2002revelation}. They show that, for any equilibrium outcome of this broader communication game, there exists an equilibrium of the menu-posting game that induces the same outcome.



Our main result introduces a sufficient condition, \textbf{density-regularity}, under which a firm's best response to any arbitrary menu profile offered by competitors must be equivalent to a posted-price menu with a unique price. This result also implies that any equilibrium, if it exists, must be equivalent to a posted-price equilibrium. We then establish the existence of a posted-price equilibrium and its uniqueness under certain conditions.

The optimality of posted prices is not a direct extension of the canonical one-dimensional screening model. Once competition and heterogeneous preferences are introduced, standard approaches in screening—such as pointwise maximization after integration by parts or extreme-point arguments over the space of incentive-compatible mechanisms—are no longer directly applicable, because each firm's expected revenue depends on its opponents' strategies. We therefore adopt an optimal-control approach. Under this approach, the density-regularity condition emerges naturally and plays a crucial role in establishing the optimality of posted prices.     

Density-regularity is closely related to Myerson's regularity condition. While Myerson's regularity requires the virtual value function $v_i(\theta_i) = \theta_i - \frac{1-F_i(\theta_i)}{f_i(\theta_i)}$ to be increasing, our sufficient condition requires \textbf{density-weighted} virtual value function $v_i(\theta_i)f_i(\theta_i)$ to be increasing. This condition holds for distributions whose densities do not decline too sharply, such as the uniform distribution, but it may fail for distributions whose densities are concentrated on a few small intervals. We show through a counterexample that, when the density has two peaks and density-regularity fails, posted prices can be suboptimal. This contrasts with the canonical Myersonian setting, in which the optimality of posted prices does not require any regularity condition.  


We also extend our analysis to a setting with non-exclusive contracts, focusing specifically on a duopoly framework. This setup naturally captures markets for divisible goods, such as coffee, meals, or gasoline, where an agent can split her fixed demand across multiple firms.



Given a menu profile $\bm{M}$, we say that a contract $(q_i, t_i) \in M_i$ is \textbf{active} if it is selected by the agent with positive probability. We show that any equilibrium profile $\bm{M}$ in the non-exclusive duopoly setting must contain an active posted-price contract. We further introduce a sufficient condition under which the posted-price contract is the unique active contract. This condition is related to density-regularity: any distribution whose density does not drop sharply will satisfy it. Moreover, our findings indicate that sustaining an equilibrium with a larger number of distinct active contracts requires a more restrictive class of distributions.

These results suggest that multi-contract menus may be redundant: in equilibrium, only the posted-price contract is selected. Alternatively, such menus may be interpreted as collections of distinct posted-price contracts targeted at consumers with heterogeneous total demands. Either interpretation supports the view that focusing on exclusive contracting is not overly restrictive, since demand splitting does not occur. We suspect that this non-splitting result is even more robust in practice, as agents likely face additional transaction costs when splitting demand—frictions that we abstract from in our framework.

In the last part of the analysis, we relax the assumption that the agent's valuations for goods from different firms are independently drawn. We show that if the agent's valuations are negatively correlated, then our main result—that the posted-price menu is the unique best response to any profile of opponents' menus—remains valid. Negative correlation captures the idea that when the agent values one firm's good more highly, she is less likely to value other firms' goods highly. The intuition is that firms enjoy greater market power under competition and behave more like monopolists selling to a loyal consumer base. 

However, once we go beyond negative correlation, the analysis becomes largely intractable. We cannot derive a sufficient condition on the distribution alone that pins down a unique structure for the best response. The best response need not be a posted-price menu, and its structure may depend on both the opponents' profile and the specific correlation structure. We then take a step back and ask whether a posted-price menu profile can at least sustain an equilibrium. We derive a sufficient condition, correlation-adjusted density-regularity, imposed solely on the distribution of types, that guarantees the existence of a posted-price equilibrium. We also show that a posted-price equilibrium exists when types are perfectly correlated.

We close the introduction with a final remark: each of the assumptions we adopt plays an essential role in establishing the main result. Although the posted-price result is standard in the canonical one-dimensional monopoly screening problem, we emphasize that considerably stronger restrictions are required to establish it as a best response to opponents' menus once competition and heterogeneous preferences are introduced. Relaxing either density-regularity or the independence assumption makes the analysis substantially more difficult, and characterizing a clean structure for the best response and equilibrium becomes elusive. The essential contribution of this paper is to construct a tractable framework for common agency with multidimensional types, and to push the boundary of when clean results—analogous to those in the standard monopoly screening problem—can be obtained.

\subsection{Related Literature}

This paper is closely related to the literature on nonlinear pricing. The pioneering contributions of \textcite{mussa1978monopoly} and \textcite{maskin1984monopoly} establish the standard framework for monopoly screening, in which a seller offers a nonlinear tariff to screen consumers with privately known willingness to pay. In the canonical vertical-screening interpretation, the contract variable $q$ represents product quality, and the agent has a scalar type $\theta$ that determines her marginal value for quality. Subsequent work extends this framework to environments with multiple firms and strategic competition. These models are closely related to menu-posting games, since firms compete by offering schedules or menus of contracts. Much of this literature preserves the one-dimensional vertical structure: firms compete over nonlinear schedules, but the agent’s marginal value for quality is governed by the same scalar type $\theta$ across firms \parencite{champsaur1989multiproduct, stole1995nonlinear, martimort1992multiprincipaux, martimort1996exclusive, stole1991mechanism, biglaiser1993principals}.

Our model differs in both the interpretation of the contract variable and the structure of heterogeneity. We interpret $q_i$ as a quantity or probability of firm $i$'s good rather than as a common quality dimension. More importantly, the agent’s marginal valuation is firm-specific: utility from firm $i$ depends on $\theta_i q_i$ rather than on a common term $\theta q_i$. Heterogeneity is therefore horizontal and multidimensional along the screening dimension itself.

There are also papers that study environments in which the agent has horizontal heterogeneity across goods offered by different firms. In these models, horizontal heterogeneity is often interpreted as a brand-specific shock, arising from factors such as transaction costs, location, or idiosyncratic preferences. A typical formulation takes the agent's utility from firm $i$ to be $q\theta-\xi_i$,
where $\xi_i$ is a firm-specific shock. In this specification, the shock affects the level of utility from choosing a firm, but the marginal value of $q$ is still governed by the common scalar type $\theta$, which is the main distinction from our model.

Some models further restrict horizontal heterogeneity to be summarized by a single index \parencite{spulber1989product, stole1995nonlinear, borenstein1985price, katz1984price, mezzetti1997common}. This approach is in the spirit of the spatial-differentiation idea introduced by \textcite{hotelling1929stability}, but later work adapts and extends the framework to richer environments with nonlinear pricing and competitive screening. \textcite{yang2008nonlinear} study a duopoly model with two-dimensional heterogeneity: each agent has a vertical taste parameter $\theta$ and a horizontal location parameter, represented by distances $(d_1,d_2)$ satisfying $d_1+d_2=1$.

\textcite{armstrong2001competitive} also study a duopoly market in which the agent's type is described by $(\theta,\xi_A,\xi_B)$, where $\xi_i$ is a firm-specific shock. Their approach formulates competition in utility space, while the mapping from consumer utility to payments is fixed. By contrast, our model allows firms to design arbitrary menus, so that the mapping from the agent's utility to transfers is endogenously chosen rather than restricted ex ante. \textcite{rochet2002nonlinear} address a closely related problem and extend the analysis to oligopoly. However, as emphasized by \textcite{stole2007price}, closed-form equilibrium solutions are generally unavailable in their model. 

The tractability obtained in \textcite{rochet2002nonlinear} relies on additional structure. Their oligopoly equilibrium characterization requires full market coverage and a condition on the inverse hazard rates of firms' winning probabilities, evaluated along the utility vector induced by each firm in the equilibrium. By contrast, our density-regularity condition is imposed directly on the primitive distribution of the agent's valuations. Because this condition is independent of the utilities induced by rival firms' equilibrium strategies, it can be used to analyze a firm's optimization problem off equilibrium. This allows us to derive a best-response characterization, rather than only an equilibrium characterization.

The recent work of \textcite{gomes2025nonlinear} adopts a related $(\theta,\xi_1,\ldots,\xi_n)$ framework to study how consumers' propensity to switch brands shapes market outcomes and equilibrium. They restrict the common vertical type $\theta$ to be binary, high or low, and analyze how differences in brand-switching propensities across consumer types affect market efficiency.

Only a few papers in the non-linear pricing competition literature study environments in which the agent's preferences exhibit multidimensional horizontal heterogeneity. Even in these settings, the existing literature primarily focuses on equilibrium characterization, market efficiency, and comparative statics, often under substantial additional structure. Our paper adopts a more general framework and derives a sharp best-response result under a primitive regularity condition on the distribution of valuations. We also study extensions with non-exclusive contracting and correlated types, two issues that have received relatively little attention in this literature.

The remainder of the paper is organized as follows. We introduce the benchmark model in \cref{sec:model}. The main result, establishing that a posted price is the optimal response under density-regularity, is presented in \cref{sec:e}. We extend our analysis to non-exclusive contracts in \cref{sec:n}. The discussion regarding correlated type distributions is presented in \cref{sec:cor}. Finally, we conclude in \cref{sec:conclude}.

 \section{Model}\label{sec:model}

There is an agent who has a fixed demand for a certain good, normalized to one. There are multiple firms that offer similar goods that are substitutable to varying degrees. The agent has different valuations for each firm's good. The agent's type is described by a vector $\bm{\theta} = (\theta_1,\dots,\theta_n) \in \mathbb{R}^n$, where $n \ge 2$ is the number of firms and $\theta_i$ represents the agent's valuation for firm $i$'s good. We assume that each $\theta_i$ is independently drawn from a cumulative distribution function $F_i$ with a continuously differentiable, strictly positive density $f_i$ on $\Theta_i = [0, \bar{\theta}_i] \subset \mathbb{R}_{\ge 0}$. The joint distribution of the agent's type $\bm\theta$ is denoted by $\bm F$, with density $\bm f$ and support $\bm\Theta$. Firms do not know the agent's exact type but know the joint distribution $\bm F$.

Firms aim to contract with the agent and compete with each other. In the most general setting, this competition induces a complicated communication game in which the agent should also be treated as a player. After observing her type, the agent sends messages to the firms. Each firm then chooses an outcome for the agent based on the messages it receives. The agent's strategy is a mapping from her type space to a message space, while each firm's strategy is a mapping from the message space to the outcome space. In settings with multiple firms, it has been shown that the revelation principle is no longer valid. There exist equilibria in the broad communication game that cannot be captured by direct mechanisms.

\textcite{peters2001common} shows that competition between firms can be effectively modeled by having firms offer menus of outcomes. That is, firms simultaneously post the sets of outcomes they are willing to implement, and then the agent chooses her preferred option. This setting is referred to as the \textbf{menu-posting game}. This result indicates that, for any equilibrium in the broader communication game, there exists an equilibrium in the menu-posting game that induces the same equilibrium outcome. Hence, we adopt the menu-posting framework to model competition among firms, which is described below.

A \textbf{contract} $(q_i, t_i)\in [0,1] \times \mathbb{R}_{\ge0}$ consists of two elements: $q_i$ is the quantity, or probability, of the good allocated to the agent, and $t_i$ is the corresponding transfer. A \textbf{menu} $M_i \subseteq [0,1] \times \mathbb{R}_{\ge0}$ is a set of contracts. Firms compete with each other by choosing menus to post. We assume that each firm can only post a menu that is a compact subset of the set of possible allocations and transfers. In addition, each posted menu must contain the \textbf{walk-away contract} $(0,0)$, ensuring that the agent always has the option not to contract with any firm. Let $\mathcal{A}$ denote the set of menus that firms can post, which is also the strategy space in the menu-posting game.

The utility of a contract $(q_i,t_i)$ offered by firm $i$ to an agent with type $\theta_i$ is defined as $q_i\theta_i-t_i$. Since $M_i$ is compact, there exists a contract in $M_i$ that maximizes the agent's utility. Thus, we define the agent's indirect utility from firm $i$'s menu $M_i$ as
\begin{equation}\label{eq:indirect}
    U_i(\theta_i)
    =
    \max_{(q_i,t_i)\in M_i}
    \bigl(q_i\theta_i-t_i\bigr).
\end{equation}
We denote by $(q_i(\theta_i),t_i(\theta_i))\in M_i$ a contract chosen by an agent with type $\theta_i$ from firm $i$'s menu. When multiple contracts in the same menu yield the same maximal utility for the agent, we assume that she selects the contract with the highest transfer to the firm. Therefore, $(q_i(\theta_i),t_i(\theta_i))\in M_i$ is well-defined.

Note that $U_i$ is convex and non-decreasing, since it is the maximum of a family of affine functions with non-negative slopes $q_i\ge 0$. By standard properties of convex functions, $U_i$ is absolutely continuous and therefore differentiable almost everywhere. Moreover, by the envelope theorem,
\begin{equation}\label{eq:envelope}
    U_i'(\theta_i)=q_i(\theta_i)
    \quad \text{almost everywhere.}
\end{equation}
Therefore, $U_i$ is strictly increasing on any interval over which $q_i(\theta_i)>0$ almost everywhere.

The menu-posting game proceeds as follows. Each firm $i$ simultaneously posts a menu $M_i$. Then, the agent's type vector $\bm{\theta}$ is realized. After observing the menu profile $\bm{M} = (M_1, \dots, M_n)$, the agent evaluates the maximum utility obtainable from each firm based on $\bm{\theta}$ and contracts with the firm that yields her the highest utility. We assume that contracting is exclusive, meaning that the agent can contract with at most one firm. When multiple firms yield the same maximal utility, the agent breaks ties according to an arbitrary fixed tie-breaking rule. Finally, the chosen contract is implemented, and the game concludes.

Note that the agent's type vector has not yet been realized when firms post their menus. Hence, the equilibrium notion is ex ante. To define expected revenue, we first introduce the interim probability that firm $i$ weakly secures the contract conditional on the agent's valuation $\theta_i$:
\[
P_i(\bm M, \theta_i)
=
\Pr\left( U_i(\theta_i) \ge \max_{j\ne i} U_j(\theta_j) \right)
=
\prod_{j\ne i}F_j\left(U_j^{-1}(U_i(\theta_i))\right),
\]
where $U_j^{-1}$ is the generalized inverse of $U_j$, defined by
\[
U_j^{-1}(u)
=
\sup\{\theta_j \in \Theta_j: U_j(\theta_j) \le u\}.
\]
The expected revenue that firm $i$ obtains under menu profile $\bm M$ is
\[
ER_i(\bm M)
=
\int_{\Theta_i} t_i(\theta_i) P_i(\bm{M}, \theta_i) f_i(\theta_i) \, d\theta_i.
\]

The expected revenue $ER_i$ is well-defined independently of the particular fixed tie-breaking rule. Under continuous type distributions, ties outside flat regions of the relevant indirect utility functions occur only on measure-zero sets. Hence, a positive-measure tie involving firm $i$ can arise only on a region where $U_i$ is flat, which implies $q_i(\theta_i)=0$ almost everywhere. Since transfers are nonnegative and the walk-away contract $(0,0)$ is included in $M_i$, we must have $t_i(\theta_i)=0$ almost everywhere on such a region. Therefore, any fixed tie-breaking rule yields the same expected revenue for firm $i$. Hence, when defining \(ER_i\), we can, without loss of generality, use the weak-winning probability in the definition of \(P_i\), rather than explicitly incorporating the tie-breaking rule into the formula.

Finally, we define the equilibrium concept. A menu profile $\bm{M}$ is an equilibrium in the menu-posting game if, for all $1 \le i \le n$, we have
\[
ER_i(\bm{M})  \ge ER_i(\tilde{M}_i, M_{-i}) 
\quad \forall \; \tilde{M}_i \in \mathcal{A}.
\]
This states that no firm $i$ can achieve a higher expected revenue by unilaterally deviating to any alternative feasible menu $\tilde{M}_i$.
    \section{Main Result}\label{sec:e}

Our main result characterizes the structure of a firm's best response. Given an arbitrary opponents' menu profile $M_{-i}$, we formulate firm $i$'s best response problem and derive a sufficient condition---\textbf{density-regularity}---under which firm $i$'s best response must induce the same outcome as a posted-price menu. That is, every best response is outcome-equivalent to a menu of the form $M_i^\star = \{(0,0),(1,p_i^\star)\}$ for some $p_i^\star >0$. We also show that density-regularity plays an essential role: without it, the best response can be highly complex and difficult to characterize.

We then prove the main theorem using optimal control theory as our primary analytical tool. The proof also constitutes a technical contribution. Although the posted-price result superficially resembles the standard result in monopoly screening, the problem we study cannot be solved using Myerson's integration-by-parts approach or an extreme point argument. The proof further reveals how density-regularity arises naturally. By comparing our problem with the monopoly screening benchmark, we provide a sharper understanding of Myerson's regularity condition through the lens of optimal control. Finally, we establish the existence of equilibrium and provide sufficient conditions for uniqueness.

\subsection{Characterization of the Best Response}

Fix an arbitrary opponents' menu profile $M_{-i}$. For each competing firm $j\ne i$, let $U_j(\theta_j)$ denote the agent's indirect utility induced by firm $j$'s menu $M_j$. We want to derive a menu $M_i^\star$ such that
\[
ER_i(M_i^\star,M_{-i}) \ge ER_i(M_i,M_{-i}),
\quad \forall\, M_i \in \mathcal{A}.
\]

For any menu $M_i$ offered by firm $i$, let $U_i(\theta_i)$ denote the induced indirect utility function defined in \eqref{eq:indirect}. Recall that $(q_i(\theta_i),t_i(\theta_i))\in M_i$ denotes an optimal contract chosen by an agent of type $\theta_i$, and that, by \eqref{eq:envelope},
\[
U_i'(\theta_i)=q_i(\theta_i)
\quad \text{almost everywhere.}
\]
By the fundamental theorem of calculus,
\[
U_i(\theta_i)
=
U_i(0)+\int_0^{\theta_i}q_i(\tau)\,d\tau.
\]
Since transfers are nonnegative and the walk-away contract $(0,0)$ is available, we have $U_i(0)=0$. Moreover, since $U_i$ is convex, the allocation rule $q_i$ is non-decreasing. Therefore, posting a menu $M_i$ is mathematically equivalent to specifying a non-decreasing allocation rule $q_i(\theta_i)\in[0,1]$. The corresponding transfer is determined by
\[
t_i(\theta_i)
=
\theta_i q_i(\theta_i)-U_i(\theta_i).
\]

Firm $i$'s optimization problem is therefore to choose a non-decreasing allocation rule $q_i(\theta_i)\in[0,1]$ to maximize expected revenue:
\[
ER_i
=
\int_{\Theta_i}
\Big[\theta_i q_i(\theta_i)-U_i(\theta_i)\Big]
P_i(U_i(\theta_i))f_i(\theta_i)\,d\theta_i,
\]
where
\[
P_i(U_i(\theta_i))
=
\prod_{j\ne i}
F_j\left(U_j^{-1}(U_i(\theta_i))\right).
\]

This maximization problem resembles the monopolist screening problem, but is substantially complicated by the presence of the term \(P_i(U_i(\theta_i))\). The allocation rule \(q_i(\theta_i)\) affects expected revenue not only through the transfer $t_i(\theta_i)$ but also through the winning probability \(P_i(U_i(\theta_i))\). Hence, changing the allocation rule $q_i$ creates a trade-off: providing more utility for the agent will lower the transfer extracted from her, but it may increase the probability that the agent chooses firm \(i\). This interaction makes the competitive screening problem substantially more intricate than the standard monopoly problem.

The standard approach of \textcite{myerson1981optimal} is not directly applicable here. The term \(P_i(U_i(\theta_i))\) prevents us from reorganizing the objective function through integration by parts and reducing the problem to pointwise maximization. The extreme point argument of \textcite{manelli2007multidimensional} is likewise unavailable, since we cannot guarantee that the objective function is quasi-convex in \(U_i\). A new regularity condition and a different proof approach are therefore needed. We introduce the regularity condition below and present the proof in the next subsection.

\begin{definition}
A density $f_i$ is density-regular if
\[
2f_i(\theta_i) + \theta_i f_i'(\theta_i) > 0
\quad \forall\, \theta_i \in \Theta_i.
\]
\end{definition}

Recall that Myerson's regularity condition requires the virtual valuation function
\[
v_i(\theta_i)
=
\theta_i - \frac{1-F_i(\theta_i)}{f_i(\theta_i)}
\]
to be increasing. In our setting, it is useful to consider the density-weighted virtual value $v_i(\theta_i)f_i(\theta_i)$. Taking the derivative with respect to $\theta_i$ yields
\begin{align*}
    \frac{d}{d\theta_i}\Big[v_i(\theta_i)f_i(\theta_i)\Big]
    &=
    \frac{d}{d\theta_i}
    \Big[\theta_i f_i(\theta_i)-1+F_i(\theta_i)\Big] =
    2f_i(\theta_i)+\theta_i f_i'(\theta_i).
\end{align*}
Thus, density-regularity is equivalent to requiring the density-weighted virtual value to be strictly increasing. This condition is satisfied by many commonly used distributions on bounded supports, including the uniform distribution and the Beta$(\alpha,1)$ family for any $\alpha\ge 1$. Intuitively, density-regularity requires that the density not decline too rapidly on its support. With the density-regularity condition established, we are now ready to present the main result of the paper.

\begin{theorem}\label{thm:br}
If $f_i$ is density-regular, then for any menu profile $M_{-i}$ posted by competing firms, there exists a unique $p_i^\star \in \Theta_i$ such that every best response of firm $i$ induces the same outcome as the posted-price menu $
M_i^\star=\{(0,0),(1,p_i^\star)\}.$
\end{theorem}

When we say that every best response induces the same outcome as $M_i^\star$, we mean something stronger than equal expected revenue. The statement says that the posted-price contract $(1,p_i^\star)$ is the only contract selected by the agent with positive probability. In particular, if firm $i$ posts a menu under which multiple contracts are selected with positive probability, that menu cannot be a best response. This result is particularly powerful: it does not merely show that a profile of posted-price menus constitutes an equilibrium, but rather implies that any equilibrium must be outcome-equivalent to a posted-price profile. Consequently, we can, without loss of generality, restrict attention to posted-price menus in the subsequent equilibrium analysis.

Under monopoly, it is well established that optimal mechanisms in multidimensional screening are often intractable, requiring complicated or stochastic contracts \parencite{rochet1998ironing,manelli2007multidimensional}. At the same time, posted-price menus are the most prevalent selling format in practice, even though consumers inevitably hold multidimensional valuations for substitutable goods across firms. This raises a natural question: do firms adopt posted prices because they maximize expected revenue, or merely because they are easy to implement? Our findings bridge the gap between the theoretical analysis and this ubiquitous empirical observation. In markets that satisfy density-regularity, competition acts as a simplifying force, collapsing a complicated multidimensional optimization problem into a practical and easily implementable menu.

A key contribution of \cref{thm:br} is that it identifies the types of markets in which competition simplifies screening. The sufficient condition, density-regularity, is satisfied when consumer valuations are smoothly and gradually distributed across the population. This condition is especially natural for essential or standardized goods such as books, electronics, movie tickets, or furniture, where willingness to pay tracks income or usage intensity without abrupt jumps. Formally, any distribution with a non-decreasing density $f_i$, such as the uniform distribution, satisfies density-regularity.

We emphasize that density-regularity is essential for establishing the result. Without it, a posted-price menu need not be a best response, as we show in the next example. This contrasts with the standard monopoly screening problem, where the posted-price result holds regardless of whether Myerson's regularity condition is satisfied. Density-regularity is more likely to fail in markets with distinct consumer segments and sharply heterogeneous valuations. Classic examples include airline passengers, where business and leisure travelers have very different willingness to pay, and software licenses, where student and professional buyers form two clearly separated groups. In both cases, the valuation distribution may be bimodal, with a pronounced valley between the two segments, creating the kind of non-monotone density that can violate density-regularity. 

\begin{example}\label{ex:two-peak}
Consider a duopoly where firm 2 posts a price $p_2 \in (0, 1)$ and 
$\theta_2 \sim \mathrm{Uniform}[0, 1]$, so that firm 1's winning probability 
is $P_1(U_1) = \Pr(\theta_2 - p_2 \le U_1) = \min\{U_1 + p_2, 1\}$. Let 
$\theta_1$ follow a two-point distribution with mass $\alpha$ at $\theta_L$ 
and $(1-\alpha)$ at $\theta_H$, where $0 < \theta_L < \theta_H$ and 
$\Delta\theta = \theta_H - \theta_L$. Although a two-point distribution is 
not $C^1$, it is the natural limit of bimodal $C^1$ densities; by continuity, 
the conclusion below extends to any sufficiently sharp $C^1$ bimodal density 
approximating it. We show that the posted-price menu is suboptimal in this 
example.

We first derive the optimal posted-price menu. If firm 1 adopts a posted 
price $p_1$, it either skims the high type or pools both types.

If firm 1 charges a skimming price $p_1 \in (\theta_L, \theta_H]$, only the 
high type purchases. Assuming $\Delta\theta + p_2 < 1$, the expected revenue is
\[
ER_{\text{skim}}(p_1) = (1-\alpha)\, p_1\,(\theta_H - p_1 + p_2).
\]
This is concave in $p_1$, with interior optimum at 
$p_1^{\text{skim}} = (\theta_H + p_2)/2$, yielding
\[
ER^*_{\text{skim}} = \frac{(1-\alpha)(\theta_H + p_2)^2}{4}.
\]
If firm 1 charges a pooling price $p_1 \in [0, \theta_L]$, both types 
purchase. The expected revenue is
\[
ER_{\text{pool}}(p_1) = p_1\bigl[\alpha(\theta_L - p_1 + p_2) 
+ (1-\alpha)(\theta_H - p_1 + p_2)\bigr].
\]
Assuming $p_2 > \theta_L$, the constrained optimum is $p_1^{\text{pool}} 
= \theta_L$, yielding
\[
ER^*_{\text{pool}} = \theta_L\bigl[(1-\alpha)\Delta\theta + p_2\bigr].
\]
The optimal posted-price revenue is $ER^*_{\text{posted}} = \max\bigl(
ER^*_{\text{skim}},\, ER^*_{\text{pool}}\bigr)$. We restrict attention to 
the parameter regime where pooling dominates skimming:
\[
\theta_L\bigl[(1-\alpha)\Delta\theta + p_2\bigr] 
\;\geq\; \frac{(1-\alpha)(\theta_H + p_2)^2}{4}.
\]
Hence $ER^*_{\text{posted}} = ER^*_{\text{pool}} = \theta_L[(1-\alpha)
\Delta\theta + p_2]$, attained at $p_1 = \theta_L$.

Now consider the alternative menu
\[
\tilde{M}_1(q_L) = \bigl\{(0,\,0),\;(q_L,\;q_L\theta_L),\;
(1,\;\theta_H - q_L\Delta\theta)\bigr\}, \quad q_L \in [0,1].
\]
Under $\tilde{M}_1(q_L)$, the low type selects $(q_L, q_L\theta_L)$ 
with $U_L = 0$, and the high type selects $(1, \theta_H - q_L\Delta\theta)$ 
with $U_H = q_L\Delta\theta$. The expected revenue is
\[
ER(q_L) = \alpha\, q_L\theta_L\, p_2 
+ (1-\alpha)(\theta_H - q_L\Delta\theta)(q_L\Delta\theta + p_2).
\]
At $q_L = 1$, both non-trivial contracts collapse to $(1, \theta_L)$, 
recovering the optimal posted price $p_1 = \theta_L$, so 
$ER(1) = ER^*_{\text{posted}}$. We now show there exists $q_L^* < 1$ 
such that $ER(q_L^*) > ER(1)$.

Direct computation gives
\[
ER''(q_L) = -2(1-\alpha)(\Delta\theta)^2 < 0,
\]
so $ER$ is strictly concave on $[0,1]$. The slope at $q_L = 1$ is
\[
ER'(1) = \alpha\theta_L p_2 - (1-\alpha)\Delta\theta\,
(\theta_H + p_2 - 2\theta_L).
\]
We further restrict to parameters satisfying
\[
\alpha\theta_L p_2 \;<\; (1-\alpha)\Delta\theta\,(\theta_H + p_2 - 2\theta_L),
\]
which guarantees $ER'(1) < 0$. By strict concavity, the optimum 
$q_L^* \in (0,1)$ is interior, with $ER(q_L^*) > ER(1)$. We conclude 
that the alternative menu $\tilde{M}_1(q_L^*)$ strictly dominates every 
posted-price menu in this parameter regime. As a concrete instance, 
$(\theta_L,\, \theta_H,\, \alpha,\, p_2) = (0.3,\, 0.7,\, 0.5,\, 0.4)$ 
satisfies all parameter constraints, yielding $q_L^* = 0.75$, 
$ER(q_L^*) = 0.185$, and $ER^*_{\text{posted}} = 0.18$.
\end{example}

\subsection{Proof of \cref{thm:br} and the Role of Density-Regularity}

\Cref{ex:two-peak} illustrates that when the type distribution concentrates 
on distinct intervals, firms have an incentive to offer multi-contract menus 
in order to extract more revenue from each segment. We refer to this scenario 
as a multi-peak market. This demonstrates how a violation of density-regularity 
can cause the posted-price result of \cref{thm:br} to break down. We emphasize 
that the multi-peak market in \cref{ex:two-peak} is only one instance of its 
failure. In general, type distributions can exhibit a rich variety of features, 
and violations of the posted-price result may arise from sources other than 
multi-peak markets. On the other hand, density-regularity is only a sufficient condition for the 
posted-price result. There may exist distributions that violate 
density-regularity for which a posted-price menu remains the best response.

To provide a better understanding of the derivation of the density-regularity condition, we walk through the proof of \cref{thm:br} in this subsection.  We also benchmark our model against the classic monopoly setting studied by \textcite{myerson1981optimal}. By demonstrating that Myerson’s standard result emerges as a special case within our broader common agency framework, we highlight the generalizability and technical contribution of our approach.

Recall that firm \(i\)'s optimization problem is
\begin{equation}\label{eq:ER}
    \max_{q_i:\Theta_i \to [0,1]}
    \int_{\Theta_i}
    \Big[ \theta_i q_i(\theta_i) - U_i(\theta_i) \Big]
    P_i(U_i(\theta_i)) f_i(\theta_i) \, d\theta_i,
\end{equation}
subject to the constraint that \(q_i\) is non-decreasing.

In the standard monopoly case, \(P_i\) is identically equal to \(1\), since there are no competing firms. The objective function can then be simplified through integration by parts, reducing the problem to pointwise maximization. Under the standard regularity condition---that the virtual function $v_i(\theta_i)$ is strictly increasing---the monotonicity constraint on \(q_i\) is automatically satisfied at the optimum.

The primary technical challenge in our problem is that the allocation rule \(q_i\) enters the winning probability \(P_i\) through the agent's indirect utility \(U_i\). Because of this nested dependence, neither Myerson's standard approach nor the extreme point argument is directly applicable. We therefore adopt an optimal control approach.

Optimal control theory is well suited to this nested structure. In a standard optimal control problem, the optimization is carried out over a continuous time interval. At each point in time, the decision maker chooses a control variable, which determines the evolution of a state variable. Both the control and the state enter the objective function. The pair consisting of the control variable and the corresponding state variable is referred to as a trajectory. Optimal control theory aims to derive the trajectory that maximizes the objective function over the entire time interval.

In our setting, the type \(\theta_i\) plays a role analogous to the time variable. As \(\theta_i\) increases, the firm chooses an allocation \(q_i(\theta_i)\), which determines the evolution of the indirect utility \(U_i(\theta_i)\). We can therefore view each pair \((q_i,U_i)\) as a trajectory along \(\theta_i\) on the interval \([0,\bar{\theta}_i]\), where \(q_i(\theta_i)\in[0,1]\) is the control variable and \(U_i(\theta_i)\) is the state variable. By the envelope theorem, the state equation is
\[
U_i'(\theta_i)=q_i(\theta_i),
\qquad
U_i(0)=0.
\]
Our goal is to identify the optimal trajectory \((q_i,U_i)\) that maximizes the objective function~\eqref{eq:ER}.

The first step is to construct the corresponding Hamiltonian:
\begin{align*}
\mathcal{H}(\theta_i)
&=
\Big[ \theta_i q_i(\theta_i) - U_i(\theta_i) \Big]
P_i(U_i(\theta_i)) f_i(\theta_i)
+
\lambda(\theta_i) q_i(\theta_i) \\
&=
q_i(\theta_i)
\Big[
\theta_i P_i(U_i(\theta_i)) f_i(\theta_i)
+
\lambda(\theta_i)
\Big]
-
U_i(\theta_i)P_i(U_i(\theta_i))f_i(\theta_i),
\end{align*}
where \(\lambda(\theta_i)\) is the costate variable. The Hamiltonian \(\mathcal{H}(\theta_i)\) can be interpreted as the local optimization problem at each type \(\theta_i\), and \(\lambda(\theta_i)\) plays the role of a continuous-time Lagrange multiplier that prices the state evolution constraint.

By Pontryagin's Maximum Principle, any optimal trajectory $(q_i^\star, U_i^\star, \lambda^\star)$ must satisfy two necessary conditions. First, at almost every $\theta_i \in \Theta_i$, the control $q_i^\star(\theta_i)$ maximizes the Hamiltonian $\mathcal{H}(\theta_i)$ pointwise. Correspondingly, we define the switching function as
\[
\Phi(\theta_i) \equiv \theta_i P_i(U_i^\star(\theta_i)) f_i(\theta_i) + \lambda^\star(\theta_i).
\]
The optimal trajectory must satisfy $q_i^\star(\theta_i) = 1$ when $\Phi(\theta_i) > 0$ and $q_i^\star(\theta_i) = 0$ when $\Phi(\theta_i) < 0$, which is also referred to as the bang-bang solution. Second, the costate evolves according to
\[
\lambda^{\star\prime}(\theta_i) = P_i(U_i^\star(\theta_i))f_i(\theta_i) - \Big[\theta_i q_i^\star(\theta_i) - U_i^\star(\theta_i)\Big] P_i'(U_i^\star(\theta_i)) f_i(\theta_i).
\]

To establish that the optimal allocation rule \(q_i^\star\) is a step function, two issues must be addressed. First, the optimal control formulation solves a relaxed problem in which \(q_i^\star\) is not required to be non-decreasing. Hence, if \(\Phi(\theta_i)\) is not non-decreasing, the bang-bang solution induced by the Hamiltonian maximization may fail to be non-decreasing. In that case, the resulting trajectory would not be feasible for the original maximization problem.

Second, even if \(\Phi(\theta_i)\) is non-decreasing, there may exist an interval \([a,b]\) such that \(\Phi(\theta_i)=0\) for all \(\theta_i\in[a,b]\). On such an interval, any \(q_i^\star(\theta_i)\in[0,1]\) satisfies the Hamiltonian maximization condition, leaving the structure of the optimal allocation rule undetermined. We show in the appendix that no such interval can exist once the monotonicity of the switching function is established. Hence, this issue does not affect the step-function structure of the optimal allocation.

To address the monotonicity issue, we differentiate the switching function with respect to \(\theta_i\) and substitute the costate equation:
\begin{align*}
    \Phi'(\theta_i)
    &=
    \Big[
    P_i(U_i(\theta_i)) f_i(\theta_i)
    + \theta_i P_i'(U_i(\theta_i))q_i(\theta_i) f_i(\theta_i)
    + \theta_i P_i(U_i(\theta_i)) f_i'(\theta_i)
    \Big]
    + \lambda'(\theta_i) \\
    &=
    \Big[
    2f_i(\theta_i) + \theta_i f_i'(\theta_i)
    \Big] P_i(U_i(\theta_i))
    + U_i(\theta_i) f_i(\theta_i) P_i'(U_i(\theta_i)).
\end{align*}
Since \(P_i \ge 0\), \(P_i' \ge 0\), \(U_i \ge 0\), and \(f_i>0\), the second term is non-negative. Therefore, a sufficient condition for \(\Phi'(\theta_i)\ge 0\) is
\begin{equation}\label{eq:density}
    2f_i(\theta_i) + \theta_i f_i'(\theta_i) > 0,
\end{equation}
which is exactly the density-regularity condition. The condition \eqref{eq:density} also implies that the density-weighted virtual value \(v_i(\theta_i)f_i(\theta_i)\) is increasing.

We claim that the monotonicity of the density-weighted virtual valuation 
\(v_i(\theta_i)f_i(\theta_i)\) is stronger than what is strictly required to solve the screening problem via the optimal control approach. To see this, it is useful to compare our setting with the standard monopoly screening problem. In the monopoly case, where \(P_i\equiv 1\) and \(P_i'\equiv 0\), the costate equation simplifies to
\[
\lambda^{\star\prime}(\theta_i)=f_i(\theta_i).
\]
Combined with the transversality condition \(\lambda^\star(\bar{\theta}_i)=0\), integration yields
\[
\lambda^\star(\theta_i)=F_i(\theta_i)-1.
\]
The switching function therefore reduces to
\[
\Phi(\theta_i)
=
\theta_i f_i(\theta_i)+F_i(\theta_i)-1
=
v_i(\theta_i)f_i(\theta_i).
\]
To conclude that the optimal allocation rule is a step function, we do not need \(\Phi(\theta_i)\) itself to be increasing; we only need it to cross zero exactly once, moving from negative to positive. Since \(f_i>0\), this single-crossing property is guaranteed by the strict monotonicity of \(v_i(\theta_i)\), which is precisely Myerson's classical regularity condition. In the optimal control formulation, the allocation rule jumps at the point where \(\Phi(\theta_i)=0\). In the monopoly case, this is exactly the type satisfying \(v_i(\theta_i)=0\), which coincides with the cutoff derived from Myerson's approach.

We adopt the strict monotonicity of the density-weighted virtual valuation \(v_i(\theta_i) f_i(\theta_i)\) because it provides a clean sufficient condition for the result in the competitive setting. Once \(P_i\) enters the switching function, it becomes difficult to characterize the precise necessary condition for \(\Phi(\theta_i)\) to be single-crossing in terms of primitive properties of \(f_i\). We view giving up the precise necessary condition in exchange for a clean and tractable assumption as a worthwhile trade-off.

While establishing the posted-price result in the monopoly screening problem does not require any regularity condition on type distribution, density-regularity plays an essential role in our setting. Without it, the best response menu can be complicated and difficult to characterize, particularly because we place no restriction on the opponents' menu profile \(M_{-i}\) in \cref{thm:br}.

\subsection{Equilibrium Analysis}

By \cref{thm:br}, any best response is outcome-equivalent to a posted-price menu, so for the purpose of equilibrium analysis we can restrict attention to the price-posting game, which is introduced as follows: The strategy 
space for each firm $i$ simplifies to $\Theta_i = [0,\bar{\theta}_i]$, where 
each $p_i \in \Theta_i$ represents the price posted by firm $i$. A strategy 
profile is therefore represented by the vector $\bm{p} = (p_1,\dots,p_n) \in 
\prod_{i=1}^n \Theta_i$. The expected revenue of firm $i$ under the profile 
$(p_i, p_{-i})$ is
\[
ER_i(p_i, p_{-i}) = p_i \int_{p_i}^{\bar{\theta}_i} 
\prod_{j \ne i} F_{j}(\theta_i - p_i + p_{j})\, f_i(\theta_i)\, d\theta_i 
= p_i \cdot Q_i(p_i, p_{-i})\footnote{Throughout this section, we extend each \(F_j\) by setting \(F_j(x)=0\) for \(x<0\) and \(F_j(x)=1\) for \(x>\bar\theta_j\).},
\]
where $Q_i(p_i, p_{-i})$ denotes the expected demand for firm $i$'s good 
under $(p_i, p_{-i})$. A strategy profile $(p_1^\star,\dots,p_n^\star)$ 
constitutes an equilibrium if, for all $1 \le i \le n$,
\[
ER_i(p_i^\star, p_{-i}^\star) \ge ER_i(p_i, p_{-i}^\star), 
\quad \forall\, p_i \in \Theta_i.
\]
We define the best response correspondence in the price-posting game as
\[
BR_i(p_{-i}) = \arg\max_{p_i \in \Theta_i} ER_i(p_i, p_{-i}).
\]
By \cref{thm:br}, under density-regularity, there exists a unique $p_i^\star$ 
maximizing firm $i$'s expected revenue for any fixed $p_{-i}$. Hence 
density-regularity ensures that $BR_i$ is single-valued, yielding the 
following corollary.

\begin{cor}\label{cor:eqm}
Assume that $f_i$ is density-regular for all $1 \le i \le n$. There exists 
a posted-price menu profile that constitutes an equilibrium in the menu-posting game.
\end{cor}

The proof is straightforward. By Berge's maximum theorem, the best response correspondence is upper hemicontinuous. We have that \(BR_i\) is single-valued and continuous. Then, we can apply Brouwer's fixed point theorem and yield 
the existence of an equilibrium in the price-posting game, which is also an 
equilibrium in the menu-posting game.

Having established existence, the next natural question is whether the 
equilibrium is unique. Uniqueness provides a precise prediction of market 
outcomes and a tractable environment for comparative statics. Hence, we introduce 
a sufficient condition for uniqueness of equilibrium in the following proposition.

\begin{proposition}\label{prop:lattice}
Assume that $f_i$ is density-regular and log-concave for all $1 \le i \le n$. There exists a posted-price menu profile that induces the unique equilibrium outcome in the menu-posting game.
\end{proposition}

The proof consists of two parts. We first apply Tarski's fixed point theorem to establish a lattice structure for the set of equilibria. We then show that the largest and smallest equilibria must coincide, yielding uniqueness.

To apply Tarski's fixed point theorem, we require the best response function to be monotone. That is, for any two opponents' price profiles \(p_{-i}\) and \(p_{-i}'\) such that \(p_{-i}\ge p_{-i}'\), it must hold that
\[
BR_i(p_{-i})\ge BR_i(p_{-i}').
\]
Intuitively, this condition says that when opponents raise their prices, a firm's best response is to weakly increase its own price.

However, this strategic complementarity does not hold automatically. Consider a duopoly where firm \(i\)'s opponent, firm \(j\), controls a substantial segment of consumers who strongly prefer \(j\)'s product. As long as firm \(j\) remains competitively active, firm \(i\) may optimally concede this segment and choose not to compete aggressively for these consumers. If firm \(j\) effectively exits the market by setting a prohibitive price, firm \(i\) may instead find it profitable to lower its price in order to capture the consumers abandoned by firm \(j\).

This potential for non-monotonic behavior explains why, in \cref{prop:lattice}, we impose log-concavity of the density \(f_i\). In the appendix, we show that log-concavity rules out such anomalies and guarantees that the best response function is monotone.

Tarski's fixed point theorem then implies not only the existence of an equilibrium but also that the set of equilibria forms a complete lattice. Hence, there exist a largest equilibrium price vector \(\bar{\bm p}\) and a smallest equilibrium price vector \(\underline{\bm p}\). To establish uniqueness, suppose for contradiction that \(\bar{\bm p}\ne \underline{\bm p}\), and let \(k\) be a firm such that
\[
\bar p_k-\underline p_k
=
\max_i(\bar p_i-\underline p_i)>0.
\]
Since both \(\bar{\bm p}\) and \(\underline{\bm p}\) are fixed points of the best response function, we must have
\[
\bar p_k-\underline p_k
=
BR_k(\bar p_{-k})-BR_k(\underline p_{-k}).
\]
However, in the appendix, we also show that
\[
BR_k(\bar p_{-k})-BR_k(\underline p_{-k})
<
\bar p_k-\underline p_k,
\]
with the assumption of density-regularity. Therefore, assuming that \(\bar{\bm p}\ne \underline{\bm p}\) will yield a contradiction. We conclude that the largest and smallest equilibria must coincide, and the equilibrium outcome in the menu-posting game is unique.

We conclude this section with a discussion of the distributional assumptions. The two conditions play distinct roles: log-concavity ensures the lattice structure of equilibria by guaranteeing monotonicity of the best response, while density-regularity ensures the posted-price structure of best responses and plays a key role in the uniqueness argument. It is worth noting that neither condition implies the other. This contrasts with the classical monopoly screening framework of \textcite{myerson1981optimal}, where the standard regularity condition---monotonicity of the virtual value---is implied by log-concavity of the density.
Nevertheless, many commonly used distributions on bounded supports satisfy both conditions simultaneously, including the uniform distribution and the Beta$(\alpha,1)$ family for any $\alpha\ge 1$.

The uniqueness result is particularly valuable for applied work. When the equilibrium is unique, the model delivers sharp comparative statics and unambiguous policy predictions. In contrast, models admitting multiple equilibria require additional equilibrium selection arguments, which can undermine predictive power. Researchers adopting posted-price competition can invoke this result to guarantee a unique equilibrium, provided they verify that the type densities are density-regular and log-concave.

    \section{Non-exclusive Contracting}\label{sec:n}

In the benchmark model, the assumption of exclusive contracting naturally captures markets for indivisible goods, such as cars, movies, or hotel rooms, where the agent must fulfill her entire demand through a single firm. In this section, we allow contracts to be non-exclusive to study markets for divisible goods, such as coffee, meals, or office supplies. In these markets, an agent can split her fixed demand across multiple firms; for instance, a consumer requiring one gallon of coffee per month could purchase half a gallon from Starbucks and the remaining half from Blue Bottle. Furthermore, in practice, we often observe firms in these markets offering multi-contract menus that go beyond simple posted prices, such as coffee shops providing different drink sizes and fast food restaurants offering options of different meal sizes. This empirical reality motivates our subsequent inquiry: can a posted-price menu profile still sustain an equilibrium when contracts are non-exclusive? For analytical tractability, we focus on the duopoly case while adopting the core environment from the previous section.

\subsection{Menus with Finite Contracts}

We maintain the notation from \cref{sec:model}, specializing to the case \(n=2\). Each firm \(i\in\{1,2\}\) offers a menu of contracts \(M_i\subseteq [0,1]\times\mathbb R_{\ge0}\). We restrict attention to menus consisting of finitely many contracts. Thus, each menu takes the form
\[
M_i=\{(q_{i\ell},t_{i\ell})\}_{\ell=1}^{k_i}
\]
for some \(k_i\le K\), where \(K\in\mathbb N\) is the maximum number of contracts that a firm can include in its menu. To ensure that the agent can walk away or contract exclusively with a single firm, we also require each menu to contain the walk-away contract \((0,0)\). Let
\[
\mathcal A^K
=
\left\{
M\subseteq [0,1]\times\mathbb R_{\ge0}:
(0,0)\in M,\ |M|\le K
\right\}
\]
denote the set of feasible finite menus.

After observing the menu profile \(\bm M=(M_1,M_2)\), an agent with type \(\bm\theta\) simultaneously chooses one contract from each menu to maximize her quasilinear utility, subject to a unit-capacity constraint:
\[
U^{\bm M}(\bm\theta)
=
\max_{\substack{(q_1,t_1)\in M_1,\\ (q_2,t_2)\in M_2}}
\left\{
\theta_1q_1+\theta_2q_2-t_1-t_2
\right\}
\quad
\text{s.t.}
\quad
q_1+q_2\le1.
\]

Let \((q_i^{\bm M}(\bm\theta),t_i^{\bm M}(\bm\theta))\) denote the contract that an agent of type \(\bm\theta\) chooses from menu \(M_i\) under profile \(\bm M\). Whenever the agent has multiple utility-maximizing feasible contract pairs, we assume that she selects one according to an arbitrary fixed measurable tie-breaking rule. This guarantees that \((q_i^{\bm M}(\bm\theta),t_i^{\bm M}(\bm\theta))\) is well-defined.\footnote{
The results below are invariant to the particular tie-breaking rule. We only require that the rule be fixed ex ante, so that for each menu profile \(\bm M\) and type \(\bm\theta\), the selected contracts \((q_i^{\bm M}(\bm\theta),t_i^{\bm M}(\bm\theta))\) are well-defined. All selected contracts and expected revenues are evaluated relative to this fixed rule.
}

A menu profile \(\bm M\) constitutes an equilibrium if, for each firm \(i\in\{1,2\}\),
\[
\int_{\bm\Theta} t_i^{\bm M}(\bm\theta)\bm f(\bm\theta)\,d\bm\theta
\ge
\int_{\bm\Theta} t_i^{(\tilde M_i,M_{-i})}(\bm\theta)\bm f(\bm\theta)\,d\bm\theta,
\quad
\forall\,\tilde M_i\in\mathcal A^K.
\]
We refer to this setting as the \textbf{duopoly non-exclusive menu-posting game}. The existence of the equilibrium follows immediately from the result in the previous section.

\begin{cor}\label{cor:non}
Assume that \(f_i\) is density-regular for both \(i\in\{1,2\}\). There exists a posted-price menu profile that constitutes an equilibrium in the duopoly non-exclusive menu-posting game.
\end{cor}

In the duopoly case, if one firm posts a posted-price menu, then any positive-quantity contract chosen from the other firm rules out simultaneously choosing the posted-price contract, since the capacity constraint would be violated. In this sense, the agent's comparison between a positive-quantity contract from one firm and the opponent's posted-price contract resembles the exclusive-contracting problem. For the rest of this subsection, we investigate further properties of equilibrium. Before presenting the main results, we introduce two useful definitions.

Given the joint distribution of types \(\bm F\) and a menu profile \(\bm M\), certain contracts may never be chosen by any agent type. Removing these latent contracts from the menus does not alter the expected payoffs of the agent or the firms. This observation motivates the following definition.

\begin{definition}
Given a menu profile \(\bm M=(M_1,M_2)\), the \textbf{active menu} of firm \(i\), denoted by \(M_i^A(\bm M)\subseteq M_i\), is the subset of contracts in \(M_i\) that are chosen by a strictly positive measure of agent types \(\bm\theta\in\bm\Theta\), given the joint distribution \(\bm F\).
\end{definition}

A contract \((q_i,t_i)\) is active under \(\bm M\) if \((q_i,t_i)\in M_i^A(\bm M)\). Since inactive contracts are never selected, they do not affect expected revenues. Hence, for equilibrium analysis, we can restrict attention to active menus.

Next, for a given menu profile $\bm{M} = (M_1, M_2)$, an agent of type $\bm{\theta}$ might select a pair of contracts that does not fulfill her total demand, meaning $q_1^{\bm{M}}(\bm{\theta}) + q_2^{\bm{M}}(\bm{\theta}) < 1$. We can characterize the class of menu profiles under which the agent always selects contracts that fulfill her exact demand. The corresponding definition is introduced below.

\begin{definition}
A menu profile \(\bm M=(M_1,M_2)\) is \textbf{demand-matching} if, for every firm \(i\in\{1,2\}\) and every active contract \((q_i,t_i)\in M_i^A(\bm M)\), there exists a complementary active contract \((q_{-i},t_{-i})\in M_{-i}^A(\bm M)\) that is chosen simultaneously by some type \(\bm\theta\) and satisfies
$q_i+q_{-i}=1$.

\end{definition}

A menu profile that is not demand-matching gives firms an incentive to revise the corresponding active contracts. By appropriately increasing both the quantity and the transfer, a firm can induce the relevant agents to choose the revised contract and thereby increase its expected revenue. Formalizing this intuition yields the following result.

\begin{proposition}\label{prop:demand}
Every equilibrium menu profile in the duopoly non-exclusive menu-posting game is demand-matching.
\end{proposition}

With \cref{prop:demand} established, we can, without loss of generality, rewrite the agent's problem for equilibrium analysis as follows:
\[
U^{\bm M}(\bm\theta)
=
\max_{\substack{(q_1,t_1)\in M_1,\\ (q_2,t_2)\in M_2}}
\left\{
\theta_1q_1+\theta_2q_2-t_1-t_2
\right\}
\quad
\text{s.t.}
\quad
q_1+q_2=1.
\]
Substituting \(q_2=1-q_1\) and defining the relative preference
\[
\Delta=\theta_1-\theta_2,
\]
the agent's indirect utility can be written as
\[
U^{\bm M}(\bm\theta)
=
\theta_2+V^{\bm M}(\Delta),
\]
where
\[
V^{\bm M}(\Delta)
=
\max_{\substack{(q_1,t_1)\in M_1,\\ (1-q_1,t_2)\in M_2}}
\left\{
\Delta q_1-t_1-t_2
\right\}.
\]
Because \(\theta_2\) enters additively as a constant, it does not affect the agent's optimization. The optimal choice is therefore driven entirely by the relative preference \(\Delta\). Accordingly, we let \(q_1^{\bm M}(\Delta)\) denote the optimal allocation from firm 1 chosen by an agent with relative preference \(\Delta\).

The random variable \(\Delta\) is distributed according to a cumulative distribution function \(H\), with density \(h\), on support \([\underline{\Delta},\overline{\Delta}]\). Since \(V^{\bm M}\) is the upper envelope of affine functions of \(\Delta\), it is convex. Consequently, \(q_1^{\bm M}(\Delta)\) is unique almost everywhere. By the Envelope Theorem, we have
\[
\frac{dV^{\bm M}(\Delta)}{d\Delta}
=
q_1^{\bm M}(\Delta)
\quad
\text{almost everywhere}.
\]
We also know that \(q_1^{\bm M}(\Delta)\) is non-decreasing because \(V^{\bm M}\) is convex.

\begin{proposition}\label{prop:finite}
For any equilibrium menu profile in the duopoly non-exclusive menu-posting game, the active menu of each firm must contain a posted-price contract.
\end{proposition}

The proof is by contradiction. Suppose that, in equilibrium, the largest active contract of firm \(i\) allocates a strictly partial quantity \(q_{ik}<1\). Firm \(i\) can then profitably deviate by introducing a posted-price contract \((1,p_i)\) with \(p_i>t_{ik}\). By choosing \(p_i\) appropriately, firm \(i\) can ensure that the agent types who previously selected \((q_{ik},t_{ik})\) strictly prefer the new contract, while all other types remain with their original choices. Under this augmented menu, firm \(i\) strictly increases its expected revenue because the deviating types now pay a higher transfer. This contradicts the assumption that the original menu profile is an equilibrium.

A natural next question is whether the posted-price contract is the only active contract in equilibrium. An affirmative answer would imply that the agent does not split her demand across multiple firms, even though she is allowed to do so. However, the uniqueness of the active contract is not guaranteed in general; an additional condition is required. We introduce this condition in the following proposition.

\begin{proposition}\label{prop:finite_unique}
Assume that $\varphi(\Delta)=\Delta h(\Delta)+2H(\Delta)$ is strictly monotone. For any equilibrium menu profile in the duopoly non-exclusive menu-posting game, the posted-price contract is the unique active contract with positive allocation for each firm.\footnote{
The walk-away contract \((0,0)\) may also be active by the demand-matching property established in \cref{prop:demand}: if an agent chooses \((1,p_i)\) from \(M_i\), she must choose \((0,0)\) from \(M_{-i}\). Since the walk-away contract induces no allocation or transfer, we omit it from the discussion of active contracts with positive allocation.
}
\end{proposition}

To provide intuition for the sufficient condition, we present a sketch of the proof. Let \(\bm M\) be an equilibrium profile whose active menus contain contracts other than posted-price contracts. Order firm 1's active contracts by allocation:
\[
0=q_0<q_1<q_2<\cdots<q_{k-1}<q_k=1,
\]
with corresponding transfers
\[
t_1^0=0,\quad t_1^1,\ldots,t_1^{k-1},\quad t_1^k=p_1.
\]
We assume \(k\ge2\), so that at least one interior contract exists. By demand-matching, established in \cref{prop:demand}, firm 2's complementary active contracts have allocations \(1-q_j\), with transfers \(t_2^j\). In particular,
\[
t_2^0=p_2
\quad
\text{and}
\quad
t_2^k=0.
\]

Since \(V^{\bm M}\) is convex and \(q_1^{\bm M}\) is non-decreasing, the relative-preference space admits a partition
\[
\underline{\Delta}
=
\Delta_0
<
\Delta_1
<
\cdots
<
\Delta_k
<
\Delta_{k+1}
=
\overline{\Delta},
\]
such that all types \(\Delta\in[\Delta_j,\Delta_{j+1}]\) choose the contract pair
\[
(q_j,t_1^j)
\quad
\text{and}
\quad
(1-q_j,t_2^j),
\]
for \(j=0,\ldots,k\). The firms' expected revenues can therefore be written as
\begin{align*}
ER_1
&=
\sum_{j=1}^{k-1}
t_1^j\bigl[H(\Delta_{j+1})-H(\Delta_j)\bigr]
+
p_1\bigl[H(\Delta_{k+1})-H(\Delta_k)\bigr],
\\
ER_2
&=
p_2\bigl[H(\Delta_1)-H(\Delta_0)\bigr]
+
\sum_{j=1}^{k-1}
t_2^j\bigl[H(\Delta_{j+1})-H(\Delta_j)\bigr].
\end{align*}

At each interior boundary type \(\Delta_j\), the agent is indifferent between adjacent contract pairs \(j-1\) and \(j\):
\[
\Delta_j q_{j-1}-t_1^{j-1}-t_2^{j-1}
=
\Delta_j q_j-t_1^j-t_2^j,
\quad
j=1,\ldots,k.
\]
Combining these indifference conditions with the first-order conditions of both firms' expected revenue functions with respect to their transfers yields
\begin{equation}\label{eq:finite_unique}
\Delta_j h(\Delta_j)+2H(\Delta_j)
=
\Delta_{j+1}h(\Delta_{j+1})+2H(\Delta_{j+1}),
\quad
j=1,\ldots,k-1.
\end{equation}
This equality cannot hold if
\[
\varphi(\Delta)=\Delta h(\Delta)+2H(\Delta)
\]
is strictly monotone. Hence, under the condition in \cref{prop:finite_unique}, no equilibrium can contain an interior active contract. The posted-price contract is therefore the unique active contract with positive allocation for each firm.

Since \(H(\Delta)\) is increasing, monotonicity of \(\varphi(\Delta)\) is guaranteed whenever \(h(\Delta)\) does not drop sharply. For instance, the condition holds when \(\Delta\) is uniformly distributed over \([\underline{\Delta},\overline{\Delta}]\).

We note that monotonicity of \(\varphi\) is only a sufficient condition. In general, as the partition of \([\underline{\Delta},\overline{\Delta}]\) becomes finer, \eqref{eq:finite_unique} becomes increasingly difficult to satisfy, since it requires \(\varphi\) to take the same value at an increasingly large set of distinct points. The takeaway is that sustaining an equilibrium with more distinct active contracts requires stronger restrictions on the type distribution. As we show in the next subsection, this observation also helps explain why uniqueness of the active contract is obtained more readily when menus are posted as continuous tariffs offering contracts for every quantity \(q_i\in[0,1]\).

Real-world markets frequently feature firms offering multiple contracts for divisible goods---for example, coffee in different cup sizes or office supplies sold in packages of varying quantities. While consumers could theoretically fulfill a fixed demand by splitting purchases across multiple vendors, \cref{prop:finite_unique} shows that, under the sufficient condition above, agents buy from a single firm via a posted-price contract. This result suggests that the multi-contract menus observed in practice are either inactive in equilibrium or designed to screen consumers with heterogeneous total demands, rather than to facilitate demand splitting across firms.

\subsection{Menus with Continuous Tariff}

In this subsection, we connect our model to the standard literature on nonlinear pricing. We assume that each firm \(i\) produces at a constant marginal cost \(c_i\) and offers a tariff \(t_i:[0,1]\to\mathbb R_{\ge0}\), mapping each quantity \(q_i\in[0,1]\) to a transfer \(t_i(q_i)\). We restrict attention to tariffs that are continuous and almost everywhere twice differentiable. A tariff acts as a menu with a continuum of contracts, offering one contract for each possible quantity \(q_i\in[0,1]\). Formally,
\[
M_i=\{(q_i,t_i(q_i)):q_i\in[0,1]\}.
\]

Upon observing the tariffs \(t_1\) and \(t_2\), an agent of type \(\bm\theta=(\theta_1,\theta_2)\) chooses a quantity pair \((q_1,q_2)\) to solve
\begin{align*}
\max_{(q_1,q_2)\in[0,1]^2}
&\left\{
\theta_1q_1+\theta_2q_2-t_1(q_1)-t_2(q_2)
\right\}
\quad
\text{s.t.}
\quad
q_1+q_2=1  \\
\equiv
\max_{q_1\in[0,1]}
&\left\{
\Delta q_1-t_1(q_1)-t_2(1-q_1)
\right\}
+\theta_2,
\end{align*}
where \(\Delta=\theta_1-\theta_2\).\footnote{
In this subsection, we impose the assumption that the agent's unit demand is exhausted, so that \(q_1+q_2=1\). This is the continuous-tariff analogue of the demand-matching property in the finite-menu case.
}

For any given tariff profile \(\bm t=(t_1,t_2)\), let
$\bm q^{\bm t}(\Delta)
=
(q_1^{\bm t}(\Delta),q_2^{\bm t}(\Delta))$ denote the optimal allocation chosen by an agent with relative preference \(\Delta\). A tariff profile \(\bm t^\star\) constitutes an equilibrium if, for each firm \(i\in\{1,2\}\),
\[
\int_{\underline{\Delta}}^{\overline{\Delta}}
\left[
t_i^\star(q_i^{\bm t^\star}(\Delta))
-
c_iq_i^{\bm t^\star}(\Delta)
\right]
h(\Delta)\,d\Delta
\ge
\int_{\underline{\Delta}}^{\overline{\Delta}}
\left[
t_i(q_i^{(t_i,t_{-i}^\star)}(\Delta))
-
c_iq_i^{(t_i,t_{-i}^\star)}(\Delta)
\right]
h(\Delta)\,d\Delta
\]
for all admissible tariffs \(t_i:[0,1]\to\mathbb R_{\ge0}\). That is, no firm can strictly increase its expected profit by unilaterally deviating from \(\bm t^\star\). We refer to this setting as the \textbf{duopoly tariff-posting game}.

If an equilibrium tariff profile contains active contracts other than a posted-price contract, then the agent splits her demand across the two firms with positive probability. That is, a positive measure of types \(\Delta\) chooses an interior allocation \(q_i^{\bm t}(\Delta)\in(0,1)\). In the following proposition, we establish a necessary condition for the agent to choose an interior allocation in equilibrium.

\begin{proposition}\label{prop:continuous}
For any equilibrium in the duopoly tariff-posting game, an agent with type \(\Delta\) chooses an interior solution \(q_1^\star(\Delta)\in(0,1)\) only if $\psi'(\Delta)=1$,
where $\psi(\Delta)=\frac{1-2H(\Delta)}{h(\Delta)}$. Consequently, if \(\psi'(\Delta)\ne1\) almost everywhere, then the agent chooses an interior allocation with probability zero.
\end{proposition}

For a tariff profile \((t_1,t_2)\) to constitute an equilibrium, three conditions must hold simultaneously: the agent's first-order condition and each firm's best-response condition. Assuming that an agent with type \(\Delta\) chooses an interior solution, these three necessary conditions jointly imply \(\psi'(\Delta)=1\). To sustain an equilibrium with multiple active contracts, there must therefore exist an interval \([a,b]\subseteq[\underline{\Delta},\overline{\Delta}]\) such that $\psi'(\Delta)=1
\quad
\text{for all } \Delta\in[a,b]$. This is a  condition that can be satisfied only by a very specific class of distributions.

A final remark on non-exclusive contracting is that an equilibrium in which the agent splits her demand requires additional conditions on the type distribution and is easily broken by perturbations. Moreover, in practice, the agent likely incurs an additional transaction cost when splitting her demand across firms, which further discourages such behavior. The results of \cref{sec:n} therefore support the conclusion that, in a broad class of environments, we can focus on exclusive contracting when analyzing competition between firms.

    \section{Correlated Types}\label{sec:cor}

A maintained assumption in our analysis is that type distributions are independent across firms. The proof of our main result relies on this assumption in an essential way. When types are correlated, an agent's valuation for firm \(i\)'s good is informative about her valuations for competing firms' goods. Hence, from firm \(i\)'s perspective, the competitive outside option generated by other firms depends directly on \(\theta_i\), not only through the utility \(U_i(\theta_i)\) that firm \(i\) offers. This additional information structure complicates the sufficient condition for the posted-price result under the optimal control approach. In this section, we explore how far our optimal control approach can be pushed in this direction, although a complete characterization remains elusive.

\subsection{Imperfect Correlation}

We now revise the benchmark model to allow for correlated types. Let
\(\bm\theta=(\theta_1,\ldots,\theta_n)\) have joint density \(\bm f\) on the full-dimensional product support $\bm\Theta=\prod_{i=1}^n[0,\bar\theta_i].$ We assume that \(\bm f\) is continuously differentiable and strictly positive on \(\bm\Theta\). Let \(f_i\) denote the marginal density of \(\theta_i\).

Given that each competing firm \(j\ne i\) induces indirect utility \(U_j\), the probability that the agent contracts with firm \(i\) when firm \(i\) delivers utility \(U_i(\theta_i)\) to type \(\theta_i\) is
\[
P_i(\theta_i,U_i(\theta_i))
=
\Pr\left(
U_i(\theta_i)\ge \max_{j\ne i}U_j(\theta_j)
\,\middle|\,
\theta_i
\right).
\]
In contrast to the independent case, this probability depends on \(\theta_i\) not only through \(U_i(\theta_i)\), but also directly through the conditioning, since the conditional distribution of \(\bm\theta_{-i}\mid\theta_i\) varies with \(\theta_i\). When types are independent, conditioning on \(\theta_i\) is uninformative about \(\bm\theta_{-i}\), so \(P_i(\theta_i,U_i(\theta_i))\) reduces to \(P_i(U_i(\theta_i))\), recovering the formulation of the benchmark model.

Firm \(i\)'s optimization problem is to choose a non-decreasing allocation rule \(q_i\) to maximize
\[
ER_i
=
\int_{\Theta_i}
\Big[
\theta_i q_i(\theta_i)-U_i(\theta_i)
\Big]
P_i(\theta_i,U_i(\theta_i))
f_i(\theta_i)\,d\theta_i.
\]

Setting up the corresponding Hamiltonian as in the proof of \cref{thm:br} and differentiating the switching function with respect to \(\theta_i\) yields
\[
\Phi'(\theta_i)
=
\big[2f_i(\theta_i)+\theta_i f_i'(\theta_i)\big]P_i
+
\theta_i\frac{\partial P_i}{\partial \theta_i}f_i(\theta_i)
+
U_i(\theta_i)\frac{\partial P_i}{\partial U_i}f_i(\theta_i),
\]
where \(P_i\) and its partial derivatives are evaluated at $(\theta_i,U_i(\theta_i))$.

The monotonicity of \(\Phi\) is what guarantees that the relaxed optimal allocation rule \(q_i^\star\) is a step function and hence that the best response of firm \(i\) is a posted-price menu. In the independent case, the second term in \(\Phi'(\theta_i)\) is absent, since \(P_i\) does not depend directly on \(\theta_i\). Density-regularity is then sufficient to guarantee that \(\Phi\) is non-decreasing. In the correlated case, however, the additional term $\theta_i\frac{\partial P_i}{\partial \theta_i}f_i(\theta_i)$ appears. Its sign depends on the correlation structure of \(\bm\theta\) and is not guaranteed to be non-negative. Consequently, the posted-price result of \cref{thm:br} may fail under arbitrary correlation.

The two directions of correlation tell different economic stories. When types are positively correlated, firms become more directly competitive: the agent is likely to value all goods highly or poorly simultaneously, making the goods closer substitutes from her perspective. When types are negatively correlated, by contrast, each firm enjoys greater market power over its own variety, since an agent who strongly values one firm's good is less likely to value its competitors' goods highly. In this case, firms behave more like local monopolists, so our main result is more likely to extend. We formalize this intuition in the following definition and establish the corresponding result.

\begin{definition}\label{def:sdcd}
The joint distribution of \(\bm{\theta}\) is \textbf{stochastically decreasing in own type} if, for each \(i\), the conditional joint distribution function \(F_{-i\mid i}(\bm{\theta}_{-i}\mid \theta_i)\) is non-decreasing in \(\theta_i\). That is,
\[
\theta_i' \ge \theta_i
\implies
F_{-i\mid i}(\bm{\theta}_{-i}\mid \theta_i')
\ge
F_{-i\mid i}(\bm{\theta}_{-i}\mid \theta_i),
\quad
\forall\,\bm{\theta}_{-i}\in\bm{\Theta}_{-i}.
\]
\end{definition}

This condition states that consumers who value firm \(i\)'s good more highly are stochastically less likely to value competitors' goods highly. Economically, it captures horizontal product differentiation: a stronger preference for one variety is associated with weaker preferences for others. The independent case is the limiting case in which the conditional distribution does not depend on \(\theta_i\), so the inequality holds with equality.


\begin{proposition}\label{prop:correlated}
Assume that type distributions of each firm are correlated. If the marginal density $f_i$ is density-regular and the joint distribution
of $\bm \theta$ is stochastically decreasing in own type, then for any menu profile $M_{-i}$ posted by 
competing firms, there exists a unique $p_i^\star \in \Theta_i$ such that 
every best response of firm $i$ induces the same outcome as the posted-price 
menu $M_i^\star = \{(0,0),(1,p_i^\star)\}$.
\end{proposition}

Imposing \cref{def:sdcd} guarantees that the term $\theta_i f_i(\theta_i)\frac{\partial P_i}{\partial \theta_i}$ appearing in \(\Phi'(\theta_i)\) is non-negative, so the original proof strategy carries through unchanged. Beyond negative correlation, however, the analysis becomes substantially less tractable, and establishing the structure of the best response under arbitrary correlation remains an open question. While a posted-price menu need not be the best response to arbitrary opponent menus once we move beyond negative correlation, we can still ask whether a posted-price menu profile constitutes an equilibrium.

In the following, we show that when competitors are restricted to posted-price menus, the best response under the setting with correlated type distributions remains tractable. The posted-price structure allows us to write the winning probability as a conditional distribution function, which leads to a clean sufficient condition under which a firm's best response to posted prices is again a posted-price menu.

Fix firm \(i\), and suppose that all firms \(j\neq i\) post prices \(p_j\). Then firm \(j\)'s indirect utility is
\[
U_j(\theta_j)=\max\{0,\theta_j-p_j\}.
\]
If firm \(i\) gives utility \(U\) to type \(\theta_i\), then firm \(i\) wins against firm \(j\) whenever
\[
U\ge \theta_j-p_j
\quad\Longleftrightarrow\quad
\theta_j\le U+p_j.
\]
Therefore, when opponents use posted prices, firm \(i\)'s interim winning probability is
\[
P_i^p(\theta_i,U)
=
\Pr\left(\theta_j\le U+p_j\ \forall j\neq i \mid \theta_i\right)
=
F_{-i|i}(U+p_{-i}\mid \theta_i),
\]
where \(F_{-i|i}(\cdot\mid\theta_i)\) is the conditional joint distribution of \(\theta_{-i}\) given \(\theta_i\), and \(U+p_{-i}\) denotes the vector \((U+p_j)_{j\neq i}\)\footnote{We extend conditional distribution functions in the same way as in \cref{sec:e} when an argument exceeds the support.}. With this notation, after setting up the Hamiltonian as usual, the derivative of the switching function is
\[
\begin{aligned}
\Phi_i'(\theta_i)
&=
\left[
2f_i(\theta_i)+\theta_i f_i'(\theta_i)
\right]
F_{-i|i}(U_i(\theta_i)+p_{-i}\mid \theta_i) \\
&\quad
+
\theta_i f_i(\theta_i)
\frac{\partial}{\partial \theta_i}
F_{-i|i}(U_i(\theta_i)+p_{-i}\mid \theta_i) \\
&\quad
+
U_i(\theta_i)f_i(\theta_i)
\frac{\partial}{\partial U_i}
F_{-i|i}(U_i(\theta_i)+p_{-i}\mid \theta_i),
\end{aligned}
\]
where
\[
\frac{\partial}{\partial U_i}
F_{-i|i}(U_i+p_{-i}\mid\theta_i)
:=
\sum_{j\ne i}
\frac{\partial F_{-i|i}}{\partial x_j}
(U_i+p_{-i}\mid\theta_i).
\]
As before, the second term is problematic, and the monotonicity of $\Phi$ cannot be established by density-regularity alone. We therefore introduce this alternative condition below.

\begin{definition}
Fix firm \(i\). The joint distribution of \(\bm\theta=(\theta_i,\theta_{-i})\) satisfies
\textbf{correlation-adjusted density-regularity} for firm \(i\) if the conditional joint distribution \(F_{-i|i}(x_{-i}\mid \theta_i)\) is continuously differentiable in \(\theta_i\), and for every \(\theta_i\in\Theta_i\) and every \(x_{-i}\in\Theta_{-i}\),
\[
\left[
2f_i(\theta_i)+\theta_i f_i'(\theta_i)
\right]
F_{-i|i}(x_{-i}\mid \theta_i)
+
\theta_i f_i(\theta_i)
\frac{\partial}{\partial \theta_i}
F_{-i|i}(x_{-i}\mid \theta_i)
\ge 0.
\]
\end{definition}

When types are positively correlated, a higher value of \(\theta_i\) shifts the conditional distribution of opponents' valuations upward, so the conditional CDF \(F_{-i|i}(x_{-i}\mid\theta_i)\) may fall with \(\theta_i\). The condition requires the density-regularity term to dominate this correlation penalty.

When opponents are restricted to posted-price menus, firm \(i\)'s winning probability reduces to a conditional CDF of the joint distribution. This simplification allows us to formulate a primitive condition on the joint distribution \(\bm F\), independent of firm \(i\)'s choice of indirect utility \(U_i\), under which the switching function \(\Phi\) is monotone. The following proposition formalizes this result.

\begin{proposition}\label{prop:cor_post}
Suppose the joint distribution of \(\bm\theta=(\theta_i,\theta_{-i})\) satisfies correlation-adjusted density regularity for firm \(i\). Then, for any posted-price menu profile posted by competing firms, there exists a unique \(p_i^\star\in\Theta_i\) such that every best response of firm \(i\) induces the same outcome as the posted-price menu $M_i^\star=\{(0,0),(1,p_i^\star)\}$.

\end{proposition}

The same fixed-point argument used in \cref{cor:eqm} yields the following corollary.

\begin{cor}\label{cor:eqm_cor}
Suppose the joint distribution of \(\bm\theta\) satisfies correlation-adjusted density regularity for all \(1\le i\le n\). Then there exists a posted-price menu profile that constitutes an equilibrium in the menu-posting game.
\end{cor}

At first glance, correlation-adjusted density-regularity may appear hard to verify, as the dependence structure among types can be arbitrarily complex. The condition becomes transparent and tractable, however, once the correlation structure is captured by a finite-dimensional parameter. The following example illustrates this.

\begin{example}
Consider the duopoly case with support \([0,1]^2\). Suppose both marginal
distributions are uniform:
\[
F_1(\theta_1)=\theta_1,
\qquad
F_2(\theta_2)=\theta_2.
\]
Let the joint distribution be given by the Farlie--Gumbel--Morgenstern copula
\[
F(\theta_1,\theta_2)
=
\theta_1\theta_2
\left[
1+\rho(1-\theta_1)(1-\theta_2)
\right],
\]
where \(\rho\in[-1,1]\). If \(\rho\le0\), then types are negatively correlated, and \cref{prop:correlated} applies. Hence, it remains to verify correlation-adjusted density-regularity for firm \(1\) when \(\rho>0\). Since the marginal density of \(\theta_1\) is uniform, we have
\[
2f_1(\theta_1)+\theta_1 f_1'(\theta_1)=2.
\]

The conditional CDF of \(\theta_2\) given \(\theta_1\) is
\[
F_{2|1}(\theta_2\mid \theta_1)
=
\frac{\partial F(\theta_1,\theta_2)/\partial \theta_1}
{f_1(\theta_1)}.
\]
Because \(f_1(\theta_1)=1\), differentiating the joint CDF gives
\begin{align*}
F_{2|1}(\theta_2\mid \theta_1)
&=
\theta_2
\left[
1+\rho(1-2\theta_1)(1-\theta_2)
\right],\\
\frac{\partial}{\partial \theta_1}
F_{2|1}(\theta_2\mid \theta_1)
&=
-2\rho \theta_2(1-\theta_2).
\end{align*}
Correlation-adjusted density-regularity for firm \(1\) requires
\[
\left[
2f_1(\theta_1)+\theta_1 f_1'(\theta_1)
\right]
F_{2|1}(\theta_2\mid\theta_1)
+
\theta_1 f_1(\theta_1)
\frac{\partial}{\partial\theta_1}
F_{2|1}(\theta_2\mid\theta_1)
\ge 0
\]
for every \((\theta_1,\theta_2)\in[0,1]^2\). Substituting the expressions above,
this condition becomes
\begin{align*}
&2\theta_2
\left[
1+\rho(1-2\theta_1)(1-\theta_2)
\right]
-
2\rho\theta_1\theta_2(1-\theta_2)
\ge 0\\
\iff\;&
2\theta_2
\left[
1+\rho(1-\theta_2)(1-3\theta_1)
\right]
\ge 0.
\end{align*}
For \(\theta_2>0\), this is equivalent to
\[
1+\rho(1-\theta_2)(1-3\theta_1)\ge0.
\]
By continuity, the same condition characterizes the boundary case \(\theta_2=0\).

If \(\rho\ge0\), the left-hand side is minimized at \((\theta_1,\theta_2)=(1,0)\), and the minimum value is \(1-2\rho\). Therefore, for positive correlation, the condition holds if and only if
\[
\rho\le \frac12.
\]

Hence, the posted-price menu profile constitutes an equilibrium under the FGM distribution for every \(\rho\in[-1,1/2]\): for \(\rho\le0\), this follows from the negative-correlation result, while for \(0<\rho\le1/2\), it follows from the calculation above, and the correlation-adjusted density-regularity is satisfied. By symmetry, the same condition also holds for firm \(2\). 

Notice that ordinary density-regularity holds for all \(\rho\in[-1,1]\), since the marginal densities are uniform. However, density-regularity alone does not guarantee any posted-price result when \(\rho>0\). This example illustrates how correlation-adjusted density-regularity extends the equilibrium analysis to a class of positively correlated environments.
\end{example}

\subsection{Perfect Correlation}

We conclude this section by examining the case that type distributions are perfectly correlated. We focus on the following perfectly correlated type structure. There exists a
common scalar type \(\theta\), drawn from a continuously differentiable
distribution \(F\) with strictly positive density \(f\) on \([0,\bar\theta]\).
The agent's valuation for firm \(i\)'s good is
\[
    \theta_i=\theta+c_i,
\]
where \(c_i\in\mathbb R_{\ge0}\) is a firm-specific constant. Hence
\(\theta_i\) is distributed on \([c_i,\bar\theta+c_i]\). We denote the marginal
density of \(\theta_i\) by \(f_i\). Since \(\theta_i=\theta+c_i\), this density
satisfies
\[
    f_i(\theta_i)=f(\theta_i-c_i),
    \qquad
    \theta_i\in[c_i,\bar\theta+c_i] \equiv \Theta_i.
\]

One might hope that the perfect-correlation structure imposes enough 
discipline to restore the result that posted-price menu is the best response to any arbitrary opponents' menu profile. The following counterexample 
shows that this is not the case: even under perfect correlation, posting 
a posted-price menu could be suboptimal under a specific opponent's menu.

\begin{example}\label{ex:cor}
    Consider a duopoly case in which $\theta_1=\theta_2=\theta$, and
    $\theta\sim U[0,1]$. Assume that firm 2 offers the following menu 
    \[
    M_2 = \{ (q,\frac{1}{2}q^2):q \in [0,1\}],
    \]
    that induces the
    following indirect utility:
    \[
        U_2(\theta)=\frac{1}{2}\theta^2.
    \]
    If firm 1 adopts a posted-price menu and posts price $p$, the agent buys
    from firm 1 if her utility from buying from firm 1 exceeds the outside
    option:
    \[
        \theta-p\ge \frac{1}{2}\theta^2.
    \]
    Let $\theta_c$ be the marginal type who is exactly indifferent. Then
    \[
        p=\theta_c-\frac{1}{2}\theta_c^2.
    \]
    Firm 1's expected revenue from a posted price is therefore
    \[
        ER_{\text{post}}(\theta_c)
        =
        \left(\theta_c-\frac{1}{2}\theta_c^2\right)(1-\theta_c).
    \]
    The first-order condition gives
    \[
        \theta_c^\star=1-\frac{1}{\sqrt{3}},
        \qquad
        p^\star=\frac{1}{3},
        \qquad
        ER_{\text{post}}^\star=\frac{1}{3\sqrt{3}}.
    \]

    Now consider the following alternative menu. Let $\delta>0$ be sufficiently
    small, and let firm 1 choose the allocation rule
    \[
        q_1(\theta)=
        \begin{cases}
            \theta+\delta, & \text{if } 0\le \theta<\frac{1}{2},\\
            1, & \text{if } \frac{1}{2}\le \theta\le 1.
        \end{cases}
    \]
    The induced indirect utility is
    \[
        U_1(\theta)=\int_0^\theta q_1(\tau)\,d\tau,
    \]
    which gives
    \[
        U_1(\theta)=
        \begin{cases}
            \frac{1}{2}\theta^2+\delta\theta,
            & \text{if } 0\le \theta<\frac{1}{2},\\
            \theta-\frac{3}{8}+\frac{\delta}{2},
            & \text{if } \frac{1}{2}\le \theta\le 1.
        \end{cases}
    \]
    For every $\theta>0$, we have
    \[
        U_1(\theta)>U_2(\theta).
    \]
    Thus firm 1 wins every type except possibly $\theta=0$, which is a
    measure-zero event and hence does not affect expected revenue.

    The corresponding transfer is
    \[
        t(\theta)=\theta q_1(\theta)-U_1(\theta).
    \]
    Therefore,
    \[
        t(\theta)=
        \begin{cases}
            \frac{1}{2}\theta^2,
            & \text{if } 0\le \theta<\frac{1}{2},\\
            \frac{3}{8}-\frac{\delta}{2},
            & \text{if } \frac{1}{2}\le \theta\le 1.
        \end{cases}
    \]
    For sufficiently small $\delta>0$, all transfers are nonnegative. Firm 1's
    expected revenue from this menu is
    \begin{align*}
        ER_{\text{menu}}
        &=
        \int_0^{1/2}\frac{1}{2}\theta^2\,d\theta
        +
        \int_{1/2}^1
        \left(\frac{3}{8}-\frac{\delta}{2}\right)d\theta \\
        &=
        \frac{1}{48}+\frac{3}{16}-\frac{\delta}{4} \\
        &=
        \frac{5}{24}-\frac{\delta}{4}.
    \end{align*}
    Since
    \[
        \frac{5}{24}>\frac{1}{3\sqrt{3}},
    \]
    there exists a sufficiently small $\delta>0$ such that
    \[
        ER_{\text{menu}}>ER_{\text{post}}^\star.
    \]
    Therefore, a posted-price menu is not firm 1's best response.
\end{example}

\Cref{ex:cor} is inspired by \textcite{jullien2000participation}. As types become perfectly correlated, the competing offers provided by opponent firms can be viewed as type-dependent outside options. Firms need not guarantee full participation across all types; in fact, they may be better off by excluding some bottom types. This is precisely the problem studied in Section 4 of \textcite{jullien2000participation}, where it is shown that the optimal mechanism may not take the form of a posted price.

This setting differs importantly from those studied in \textcite{dworczak2024mechanism} and \textcite{augias2025economics}, where the type-dependent outside option is treated as a hard participation constraint that the principal must satisfy while solving the optimization problem. Their findings support the conclusion that, under type-dependent participation constraints, the posted price remains optimal under Myerson's regularity condition. This contrast delivers the message that competition among firms introduces complications well beyond those of the standard monopoly screening problem: because each firm has the freedom to endogenously choose which types to exclude, the resulting optimization problem becomes substantially more difficult to solve.

As before, when posted-price menus do not constitute the unique best response structure, we ask whether a posted-price menu profile can nonetheless sustain an equilibrium. Fortunately, this can be established under density-regularity.

\begin{proposition}\label{prop:per_cor}
Assume that types are perfectly correlated and that \(\theta_i=\theta+c_i\) for all \(1\le i\le n\). If each marginal density \(f_i\) is density-regular, then for any posted-price menu profile posted by competing firms such that $p_j>c_j-c_i \
\text{for all } j\ne i$, there exists a unique \(p_i^\star\in\Theta_i\) such that every best response of firm \(i\) induces the same outcome as the posted-price menu $M_i^\star=\{(0,0),(1,p_i^\star)\}$.

\end{proposition}

The condition \(p_j>c_j-c_i\) for all \(j\ne i\) ensures that firm \(i\) can win the contract with positive probability. If \(p_j\le c_j-c_i\) for some \(j\ne i\), then firm \(i\) cannot strictly beat firm \(j\) even by giving the good away for free. In that case, firm \(i\)'s maximal attainable revenue may be zero, and the best response need not be unique.

The existence of a posted-price equilibrium is guaranteed: any profile in which at least two firms $i$ post sufficiently low price such that no firms can earn positive revenue constitutes an equilibrium\footnote{Formally, let $c_i = \max_{1 \le j \le n}c_j$. Any price no larger than $c_i - \max_{j \ne i} c_j$ will be low enough to prevent firms from earning positive expected revenue regardless of any unilateral deviation.}.
Nevertheless, \cref{prop:per_cor} is not a vacuous result. In the monopoly screening problem, the result that a posted-price menu is optimal can be established without any regularity condition. Under competition, however, this is no longer the case, as the following example demonstrates.

\begin{example}
Consider a duopoly with perfectly correlated types: \(\theta_1=\theta_2=\theta\). Let the type distribution be the two-point distribution
\[
\theta\in\{\theta_L,\theta_H\}=\{0.4,1\},
\qquad
\Pr(\theta=0.4)=\Pr(\theta=1)=\frac12.
\]
Assume that firm \(2\) posts the price \(p_2=0.7\). Thus the agent's outside utility from firm \(2\) is
\[
U_2(\theta)=(\theta-0.7)_+.
\]

We first compute the best posted-price response of firm \(1\). Suppose firm \(1\) posts price \(p\). If \(p<0.4\), both types strictly prefer firm \(1\), so firm \(1\)'s revenue is
\[
ER_{\mathrm{post}}(p)=p.
\]
Hence the supremum pooling posted-price revenue is \(0.4\).

If \(0.4<p<0.7\), only the high type buys from firm \(1\), so
\[
ER_{\mathrm{post}}(p)=\frac12p<0.35<0.4.
\]
If \(p=0.7\), the high type is indifferent between the two firms. Under any fixed tie-breaking rule, firm \(1\)'s revenue is at most
\[
\frac12\cdot 0.7=0.35<0.4.
\]
If \(p>0.7\), firm \(1\) loses the high type to firm \(2\), and the low type does not buy from firm \(1\). Therefore, the supremum revenue that firm \(1\) can obtain from posted-price menus is \(0.4\).

Now consider the following non-posted menu for firm \(1\):
\[
M_1=
\left\{
(0,0),
\left(0.5,0.2-\varepsilon\right),
\left(1,0.7-2\varepsilon\right)
\right\},
\]
where \(0<\varepsilon<1/30\). The low type strictly chooses \((0.5,0.2-\varepsilon)\), obtaining utility \(\varepsilon>0\), while the high type strictly chooses \((1,0.7-2\varepsilon)\), obtaining utility \(0.3+2\varepsilon\). Both types therefore strictly prefer firm \(1\) to firm \(2\).

Firm \(1\)'s expected revenue from this menu is
\[
ER_{\mathrm{menu}}
=
\frac12(0.2-\varepsilon)
+
\frac12(0.7-2\varepsilon)
=
0.45-\frac32\varepsilon.
\]
Since \(0<\varepsilon<1/30\),
\[
0.45-\frac32\varepsilon>0.4.
\]
Thus firm \(1\)'s best response to firm \(2\)'s posted-price menu is not a posted-price menu.

This example is written as a two-point distribution for transparency. It can be approximated by a smooth density with two sharp peaks around \(0.4\) and \(1\), with a small positive density between the peaks. In the low-density region between the two peaks, \(f(\theta)\) can be made arbitrarily small while \(1-F(\theta)\) remains bounded away from zero. The virtual value falls sharply on that region, and hence the density-regularity fails. Since the revenue gap above is strict, the same strict inequality survives for all sufficiently close smooth approximations of the two-point distribution. We conclude that the best response to an opponent's posted-price menu need not be a posted-price menu when density-regularity is violated.
\end{example}

A final remark of this section is that allowing for correlated types can make each firm's best response considerably more complex. \textcite{jullien2000participation} solves the extreme case in which the two goods are treated as perfect substitutes, and even in this case the optimal mechanism can already take a complicated form. In general, the correlation structure could be richer, and how this additional information enters each firm's optimization problem is likely to be intractable in closed form. We consider this an open question worthy of future study.

    \section{Concluding Remarks}\label{sec:conclude}

This paper develops a tractable model of common agency with multidimensional types. The model is general enough to capture markets in which firms offer differentiated varieties of a substitutable good for two reasons. First, we abstract away from restrictive utility assumptions, requiring only that the agent's valuations for competing goods are drawn independently from given distributions. Second, the menu-posting game involves no loss of generality, as it can replicate any equilibrium outcome of the broader and more complex communication game inherent in common agency problems.

While previous literature suggests that multidimensional screening should involve complex mechanisms, posted pricing remains the most prevalent pricing scheme in real-world markets. Our paper bridges the gap between theory and empirical reality by identifying a sufficient condition---density-regularity---under which competition simplifies the screening problem, making the posted-price menu the unique best response to any menu posted by competitors. Beyond its theoretical interest, this result offers practical guidance to firms: whether posted pricing is optimal depends on whether the market in which the firm operates satisfies density-regularity. Conversely, if a firm does not adopt posted pricing in a market satisfying density-regularity, this may suggest that the firm perceives itself as enjoying substantial market power---an observation with potential implications for antitrust regulation. We believe this application of our result is a promising direction for future research, with the potential to connect more closely to real-world markets and to guide both business strategy and regulatory policy.

Some limitations of our study are also worth exploring further from a theoretical perspective. First, we restrict attention to the duopoly case when analyzing non-exclusive contracting. In this setting, the only strategically relevant variable is the difference in the agent's valuations for the two goods, which allows us to reduce the problem to one-dimensional screening. Introducing more than two firms precludes this dimensional reduction and makes the multidimensional screening problem substantially more difficult. Extending the analysis to more than two firms would be challenging but valuable for assessing the robustness of our results.

Second, when allowing type distributions to be correlated, our main result extends only to the case of negative correlation. We expect the analysis to become largely intractable once arbitrary correlation in the type distribution is allowed. It would be valuable to develop a framework that captures empirically realistic patterns of correlation while still delivering positive results---a direction we leave to future research.

\nocite{}  
\printbibliography 

\appendix
\section{Omitted Proof of \cref{sec:e}}

\subsection{Proof of \cref{thm:br}}

We prove \cref{thm:br}. Under the maintained assumption that transfers are nonnegative and that every menu contains the walk-away contract \((0,0)\), the induced indirect utility satisfies
\[
U_i(0)=0.
\]
Moreover, by the envelope theorem,
\[
U_i'(\theta_i)=q_i(\theta_i)
\quad \text{a.e.}
\]
Thus firm \(i\)'s problem can be written as choosing an allocation rule \(q_i:\Theta_i\to[0,1]\), with associated state variable \(U_i\), to maximize
\[
ER_i
=
\int_{\Theta_i}
\Big[
\theta_iq_i(\theta_i)-U_i(\theta_i)
\Big]
P_i(U_i(\theta_i))f_i(\theta_i)\,d\theta_i,
\]
subject to
\[
U_i'(\theta_i)=q_i(\theta_i),
\qquad
U_i(0)=0.
\]

We first solve the relaxed optimal control problem in which the monotonicity constraint on \(q_i\) is ignored. We then show that the optimal control obtained from the relaxed problem is monotone under density-regularity, and hence feasible for the original problem.

Let \(\lambda(\theta_i)\) denote the costate variable. The Hamiltonian is
\begin{align*}
\mathcal H(\theta_i)
&=
\Big[
\theta_i q_i(\theta_i)-U_i(\theta_i)
\Big]
P_i(U_i(\theta_i))f_i(\theta_i)
+
\lambda(\theta_i)q_i(\theta_i) \\
&=
q_i(\theta_i)
\Big[
\theta_iP_i(U_i(\theta_i))f_i(\theta_i)
+
\lambda(\theta_i)
\Big]
-
U_i(\theta_i)P_i(U_i(\theta_i))f_i(\theta_i).
\end{align*}

Let \((q_i^\star,U_i^\star,\lambda^\star)\) be an optimal trajectory of the relaxed problem. By Pontryagin's Maximum Principle, the optimal control maximizes the Hamiltonian pointwise almost everywhere. Define the switching function
\[
\Phi(\theta_i)
\equiv
\theta_iP_i(U_i^\star(\theta_i))f_i(\theta_i)
+
\lambda^\star(\theta_i).
\]
Since the Hamiltonian is linear in \(q_i\), the optimal control satisfies
\[
\Phi(\theta_i)>0
\implies
q_i^\star(\theta_i)=1,
\qquad
\Phi(\theta_i)<0
\implies
q_i^\star(\theta_i)=0.
\]
The costate equation is
\[
\lambda^{\star\prime}(\theta_i)
=
P_i(U_i^\star(\theta_i))f_i(\theta_i)
-
\Big[
\theta_iq_i^\star(\theta_i)-U_i^\star(\theta_i)
\Big]
P_i'(U_i^\star(\theta_i))f_i(\theta_i),
\]
with transversality condition
\[
\lambda^\star(\bar\theta_i)=0.
\]

We now show that the switching function is non-decreasing. Differentiating \(\Phi\) and substituting the costate equation gives
\begin{align}
\Phi'(\theta_i)
&=
\Big[
P_i(U_i(\theta_i))f_i(\theta_i)
+
\theta_iP_i'(U_i(\theta_i))q_i(\theta_i)f_i(\theta_i)
+
\theta_iP_i(U_i(\theta_i))f_i'(\theta_i)
\Big]
+
\lambda'(\theta_i) \notag \\
&=
\Big[
2f_i(\theta_i)+\theta_i f_i'(\theta_i)
\Big]
P_i(U_i(\theta_i))
+
U_i(\theta_i)f_i(\theta_i)P_i'(U_i(\theta_i)).
\label{eq:phi_prime}
\end{align}
Since \(P_i\ge0\), \(P_i'\ge0\), \(U_i\ge0\), and \(f_i>0\), the second term is nonnegative. Density-regularity,
\[
2f_i(\theta_i)+\theta_i f_i'(\theta_i)>0
\quad
\forall\,\theta_i\in\Theta_i,
\]
therefore implies that
\[
\Phi'(\theta_i)\ge0.
\]
Hence \(\Phi\) is non-decreasing.

It remains to address the possibility that \(\Phi\) is equal to zero on a non-degenerate interval. On such an interval, the Hamiltonian maximization condition does not uniquely determine \(q_i^\star\). The following lemma rules out any payoff-relevant singular interval.

\begin{lemma}\label{lemma:no_singular_arc}
Assume that \(f_i\) is density-regular. Let \((q_i^\star,U_i^\star,\lambda^\star)\) be an optimal trajectory. There cannot exist an interval \([\theta_a,\theta_b]\subseteq\Theta_i\), with \(\theta_b>\theta_a\), such that
\[
\Phi^\star(\theta_i)=0
\quad
\text{and}
\quad
U_i^\star(\theta_i)>0
\]
for all \(\theta_i\in[\theta_a,\theta_b]\).
\end{lemma}

\begin{proof}
Suppose, toward a contradiction, that such an interval exists. Since \(\Phi^\star\equiv0\) on \([\theta_a,\theta_b]\), we have
\[
\Phi^{\star\prime}(\theta_i)=0
\quad
\text{a.e. on }[\theta_a,\theta_b].
\]
Using \eqref{eq:phi_prime},
\[
0
=
\Big[
2f_i(\theta_i)+\theta_i f_i'(\theta_i)
\Big]
P_i(U_i^\star(\theta_i))
+
U_i^\star(\theta_i)f_i(\theta_i)P_i'(U_i^\star(\theta_i))
\]
almost everywhere on this interval.

The first bracket is strictly positive by density-regularity. Moreover, since \(U_i^\star(\theta_i)>0\), we have
\[
P_i(U_i^\star(\theta_i))>0.
\]
Indeed, for each opponent \(j\ne i\), nonnegative transfers and the walk-away contract imply \(U_j(0)=0\). Since \(U_j\) is continuous and \(U_i^\star(\theta_i)>0\), a positive-measure set of sufficiently low \(\theta_j\)'s satisfies
\[
U_j(\theta_j)<U_i^\star(\theta_i).
\]
Because each type distribution has full support, the probability that all opponents' utilities are below \(U_i^\star(\theta_i)\) is strictly positive. Thus \(P_i(U_i^\star(\theta_i))>0\).

The second term is nonnegative because \(U_i^\star(\theta_i)>0\), \(f_i(\theta_i)>0\), and \(P_i'\ge0\). Therefore,
\[
\Phi^{\star\prime}(\theta_i)>0
\]
almost everywhere on \([\theta_a,\theta_b]\), contradicting \(\Phi^\star\equiv0\) on this interval.
\end{proof}

Since \(\Phi^\star\) is non-decreasing, the set on which \(\Phi^\star=0\) is an interval, possibly degenerate. By \cref{lemma:no_singular_arc}, this interval cannot contain a non-degenerate subinterval on which \(U_i^\star>0\). On any interval where \(U_i^\star=0\), the envelope condition implies \(q_i^\star=0\) almost everywhere. Hence the Hamiltonian maximization condition implies
\[
q_i^\star(\theta_i)
=
\mathbf 1\{\Phi^\star(\theta_i)>0\}
\quad
\text{a.e.}
\]
Since \(\Phi^\star\) is non-decreasing, there exists a cutoff \(\theta_i^\star\in\Theta_i\) such that
\[
q_i^\star(\theta_i)
=
\begin{cases}
0, & \theta_i<\theta_i^\star,\\
1, & \theta_i\ge \theta_i^\star,
\end{cases}
\quad
\text{a.e.}
\]
Thus the relaxed optimal allocation rule is a step function. Since this allocation rule is non-decreasing, it is feasible for the original problem.

Finally, because \(U_i(0)=0\), the envelope condition yields
\[
U_i^\star(\theta_i)
=
\int_0^{\theta_i}q_i^\star(\tau)\,d\tau
=
\max\{\theta_i-\theta_i^\star,0\}.
\]
The corresponding transfer is
\[
t_i^\star(\theta_i)
=
\theta_iq_i^\star(\theta_i)-U_i^\star(\theta_i)
=
\begin{cases}
0, & \theta_i<\theta_i^\star,\\
\theta_i^\star, & \theta_i\ge\theta_i^\star.
\end{cases}
\]
Therefore, every optimal menu induces the same outcome as the posted-price menu
\[
M_i^\star=\{(0,0),(1,\theta_i^\star)\}.
\]
This proves that every best response is outcome-equivalent to a posted-price menu.

\subsection{Proof of \cref{cor:eqm}}

Note that the expected revenue function \(ER_i\) is continuous and that the strategy space \([0,\bar{\theta}_i]\) is compact. Hence, by Berge's Maximum Theorem, the best response correspondence \(BR_i\) is upper hemicontinuous. Combining this with the result from \cref{thm:br} that, under density-regularity, \(BR_i\) is single-valued, we conclude that \(BR_i\) is a continuous function.

Now define
\[
BR(\bm p)
=
\bigl(BR_1(p_{-1}),\dots,BR_n(p_{-n})\bigr).
\]
Since each \(BR_i\) is continuous, \(BR\) is a continuous function from the compact and convex set \(\prod_{i=1}^n[0,\bar{\theta}_i]\) into itself. By Brouwer's Fixed Point Theorem, there exists \(\bm p^\star\) such that
\[
\bm p^\star=BR(\bm p^\star).
\]
Thus, there exists an equilibrium in the price-posting game. By \cref{thm:br}, posting the menu $M_i^\star = \{(0,0), (1,p^\star_i)\}$ remains a best response to the posted-price menu profile associated with the opponents' price vector $p^\star_{-i}$ in the menu-posting game. Hence, this price-posting equilibrium $\bm p^\star$ also constitutes an equilibrium in the menu-posting game.

\subsection{Proof of \cref{prop:lattice}}

The proof of \cref{prop:lattice} consists of two parts. We first establish \cref{lemma:supermodular}, which allows us to apply Tarski's fixed point theorem and show that the set of equilibria forms a complete lattice. We then prove by contradiction that the largest and smallest equilibria must coincide.

\begin{lemma}\label{lemma:supermodular}
Assume that \(f_k\) is density-regular and log-concave for all \(1\le k\le n\). For each firm \(i\), the log-revenue function \(\ln ER_i(p_i,p_{-i})\) has non-decreasing differences in \((p_i,p_j)\) for all \(j\ne i\), on the region where \(ER_i(p_i,p_{-i})>0\).
\end{lemma}

\begin{proof}
Firm \(i\)'s expected revenue is
\[
ER_i(p_i,p_{-i})=p_i Q_i(p_i,p_{-i}),
\]
where
\[
Q_i(p_i,p_{-i})
=
\int_{p_i}^{\bar\theta_i}
f_i(\theta_i)
\prod_{k\ne i}F_k(\theta_i-p_i+p_k)
\,d\theta_i.
\]
On the region where \(ER_i>0\), we can take logarithms:
\[
\ln ER_i(p_i,p_{-i})
=
\ln p_i+\ln Q_i(p_i,p_{-i}).
\]
Since \(\ln p_i\) does not depend on any competitor's price \(p_j\), it suffices to show
\[
\frac{\partial^2 \ln Q_i(p_i,p_{-i})}{\partial p_i\partial p_j}\ge0.
\]

We extend \(f_i(x)=0\) outside \(\Theta_i\), and extend each \(F_k\) by setting \(F_k(x)=0\) below its support and \(F_k(x)=1\) above its support. With this convention, after the change of variables \(v=\theta_i-p_i\), expected demand can be written as
\[
Q_i(p_i,p_{-i})
=
\int_0^\infty
f_i(v+p_i)
\prod_{k\ne i}F_k(v+p_k)
\,dv.
\]
Let
\[
G(v,p_i,p_{-i})
=
f_i(v+p_i)
\prod_{k\ne i}F_k(v+p_k).
\]
Then
\[
Q_i(p_i,p_{-i})=\int_0^\infty G(v,p_i,p_{-i})\,dv.
\]
Hence
\[
g(v,p_i,p_{-i})
=
\frac{G(v,p_i,p_{-i})}{Q_i(p_i,p_{-i})}
\]
is a probability density over \(v\in[0,\infty)\).

Differentiating \(Q_i\) with respect to \(p_j\), for \(j\ne i\), gives
\begin{align*}
\frac{\partial Q_i}{\partial p_j}
&=
\int_0^\infty
f_i(v+p_i)
f_j(v+p_j)
\prod_{k\ne i,j}F_k(v+p_k)
\,dv \\
&=
\int_0^\infty
\frac{f_j(v+p_j)}{F_j(v+p_j)}
G(v,p_i,p_{-i})
\,dv.
\end{align*}
Therefore,
\[
\frac{\partial \ln Q_i}{\partial p_j}
=
\int_0^\infty
\frac{f_j(v+p_j)}{F_j(v+p_j)}
g(v,p_i,p_{-i})
\,dv.
\]
Define
\[
w(v,p_j)
=
\frac{f_j(v+p_j)}{F_j(v+p_j)}.
\]
Then
\[
\frac{\partial \ln Q_i}{\partial p_j}
=
\mathbb E_g[w(v,p_j)].
\]

We now show that this expectation is non-decreasing in \(p_i\). Since \(f_j\) is log-concave, \(F_j\) is log-concave by Pr\'ekopa's theorem. Hence \(\ln F_j\) is concave, and therefore
\[
w(v,p_j)
=
\frac{\partial}{\partial v}\ln F_j(v+p_j)
\]
is non-increasing in \(v\).

Next, we show that \(g(v,p_i,p_{-i})\) shifts toward lower values of \(v\) as \(p_i\) increases. It is enough to verify reverse MLRP in \(p_i\). Since
\[
\ln g(v,p_i,p_{-i})
=
\ln G(v,p_i,p_{-i})
-
\ln Q_i(p_i,p_{-i}),
\]
and \(Q_i\) does not depend on \(v\), we have
\[
\frac{\partial^2 \ln g}{\partial v\partial p_i}
=
\frac{\partial^2 \ln G}{\partial v\partial p_i}.
\]
Moreover,
\[
\ln G(v,p_i,p_{-i})
=
\ln f_i(v+p_i)
+
\sum_{k\ne i}\ln F_k(v+p_k).
\]
Only the first term depends on \(p_i\), so
\[
\frac{\partial^2 \ln g}{\partial v\partial p_i}
=
\frac{\partial^2}{\partial v\partial p_i}
\ln f_i(v+p_i)
\le0,
\]
where the inequality follows from log-concavity of \(f_i\). Hence \(g(v,p_i,p_{-i})\) satisfies reverse MLRP with respect to \(p_i\): as \(p_i\) increases, \(g\) shifts toward lower values of \(v\).

Since \(w(v,p_j)\) is non-increasing in \(v\), and \(g(v,p_i,p_{-i})\) shifts toward lower values of \(v\) as \(p_i\) increases, the expectation \(\mathbb E_g[w(v,p_j)]\) is non-decreasing in \(p_i\). Therefore,
\[
\frac{\partial^2 \ln Q_i}{\partial p_i\partial p_j}\ge0.
\]
Thus \(\ln ER_i\) has non-decreasing differences in \((p_i,p_j)\).
\end{proof}

We are now ready to apply Tarski's fixed point theorem. For each firm \(i\), define the best response function
\[
BR_i(p_{-i})
=
\arg\max_{p_i\in\Theta_i}ER_i(p_i,p_{-i}).
\]
By \cref{thm:br}, under density-regularity, \(BR_i\) is single-valued.

Since \(\ln(\cdot)\) is strictly increasing, maximizing \(ER_i\) is equivalent to maximizing \(\ln ER_i\) on the region where \(ER_i>0\). We use the convention that \(\ln ER_i=-\infty\) whenever \(ER_i=0\). Hence,
\[
BR_i(p_{-i})
=
\arg\max_{p_i\in\Theta_i}\ln ER_i(p_i,p_{-i}).
\]

By \cref{lemma:supermodular}, \(\ln ER_i\) has non-decreasing differences in \((p_i,p_j)\) for every \(j\ne i\). Since \(\Theta_i\) is a compact complete lattice, Topkis's theorem implies that \(BR_i(p_{-i})\) is non-decreasing in \(p_{-i}\). That is, if \(p_{-i}\ge p_{-i}'\), then
\[
BR_i(p_{-i})\ge BR_i(p_{-i}').
\]

Define the best response operator
\[
BR(\bm p)
=
\bigl(
BR_1(p_{-1}),\dots,BR_n(p_{-n})
\bigr).
\]
Since each \(BR_i\) is non-decreasing, \(BR\) is a monotone function from the complete lattice $\prod_{i=1}^n\Theta_i$ into itself. By Tarski's fixed point theorem, \(BR\) has a fixed point. Hence, an equilibrium exists. Moreover, Tarski's theorem implies that the set of fixed points, and therefore the set of equilibria, forms a complete lattice.

It remains to establish uniqueness of equilibrium. To do so, we first prove a contraction property of the best response in \cref{lemma:contraction}.

\begin{lemma}\label{lemma:contraction}
Assume that \(f_k\) is density-regular and log-concave for all \(1\le k\le n\). At every interior differentiability point of \(BR_i\),
\[
\sum_{j\ne i}
\frac{\partial BR_i}{\partial p_j}
<1.
\]
\end{lemma}

\begin{proof}
Firm \(i\)'s expected revenue is
\[
ER_i(p_i,p_{-i})
=
p_iQ_i(p_i,p_{-i}),
\]
where
\[
Q_i(p_i,p_{-i})
=
\int_{p_i}^{\bar\theta_i}
\prod_{k\ne i}F_k(\theta_i-p_i+p_k)f_i(\theta_i)\,d\theta_i.
\]
For \(j\ne i\), define
\[
Q_{ij}
=
\frac{\partial Q_i}{\partial p_j}
=
\int_{p_i}^{\bar\theta_i}
f_j(\theta_i-p_i+p_j)
\prod_{k\ne i,j}F_k(\theta_i-p_i+p_k)
f_i(\theta_i)
\,d\theta_i
\ge0.
\]

We first compute the own-price derivative of demand. Let
\[
g(\theta_i,p_i)
=
\prod_{k\ne i}F_k(\theta_i-p_i+p_k)f_i(\theta_i).
\]
By Leibniz's rule,
\[
\frac{\partial Q_i}{\partial p_i}
=
-g(p_i,p_i)
+
\int_{p_i}^{\bar\theta_i}
\frac{\partial g(\theta_i,p_i)}{\partial p_i}
\,d\theta_i.
\]
The boundary term is
\[
g(p_i,p_i)
=
\prod_{k\ne i}F_k(p_k)f_i(p_i)
\equiv C.
\]
Moreover,
\[
\frac{\partial g(\theta_i,p_i)}{\partial p_i}
=
-
\sum_{j\ne i}
f_j(\theta_i-p_i+p_j)
\prod_{k\ne i,j}F_k(\theta_i-p_i+p_k)
f_i(\theta_i).
\]
Therefore,
\[
\frac{\partial Q_i}{\partial p_i}
=
-C-\sum_{j\ne i}Q_{ij}.
\]

At an interior best response, the first-order condition is
\[
\frac{\partial ER_i}{\partial p_i}
=
Q_i+p_i\frac{\partial Q_i}{\partial p_i}
=
0.
\]
At an interior differentiability point of \(BR_i\), the implicit function theorem gives
\[
\frac{\partial BR_i}{\partial p_j}
=
-
\frac{
\frac{\partial^2 ER_i}{\partial p_i\partial p_j}
}{
\frac{\partial^2 ER_i}{\partial p_i^2}
}.
\]
Let
\[
N
=
\sum_{j\ne i}
\frac{\partial^2 ER_i}{\partial p_i\partial p_j},
\qquad
D
=
\frac{\partial^2 ER_i}{\partial p_i^2}.
\]
Then
\[
\sum_{j\ne i}
\frac{\partial BR_i}{\partial p_j}
=
-\frac{N}{D}.
\]

We compute \(D\) in terms of \(N\). First,
\[
\frac{\partial^2 ER_i}{\partial p_i\partial p_j}
=
Q_{ij}
+
p_i
\frac{\partial^2 Q_i}{\partial p_i\partial p_j},
\]
so
\[
N
=
\sum_{j\ne i}Q_{ij}
+
p_i
\sum_{j\ne i}
\frac{\partial^2 Q_i}{\partial p_i\partial p_j}.
\]
Next,
\[
D
=
2\frac{\partial Q_i}{\partial p_i}
+
p_i\frac{\partial^2 Q_i}{\partial p_i^2}.
\]
Differentiating
\[
\frac{\partial Q_i}{\partial p_i}
=
-C-\sum_{j\ne i}Q_{ij}
\]
with respect to \(p_i\) yields
\[
\frac{\partial^2 Q_i}{\partial p_i^2}
=
-\frac{\partial C}{\partial p_i}
-
\sum_{j\ne i}
\frac{\partial^2 Q_i}{\partial p_i\partial p_j}.
\]
Substituting into \(D\), we obtain
\begin{align*}
D
&=
2\left(-C-\sum_{j\ne i}Q_{ij}\right)
+
p_i
\left(
-\frac{\partial C}{\partial p_i}
-
\sum_{j\ne i}
\frac{\partial^2 Q_i}{\partial p_i\partial p_j}
\right) \\
&=
-N
-
\sum_{j\ne i}Q_{ij}
-
\left(
2C+p_i\frac{\partial C}{\partial p_i}
\right).
\end{align*}
Since
\[
C
=
\prod_{k\ne i}F_k(p_k)f_i(p_i),
\]
we have
\[
2C+p_i\frac{\partial C}{\partial p_i}
=
\prod_{k\ne i}F_k(p_k)
\left[
2f_i(p_i)+p_i f_i'(p_i)
\right].
\]
By density-regularity, this term is strictly positive whenever \(\prod_{k\ne i}F_k(p_k)>0\). Also, \(\sum_{j\ne i}Q_{ij} \ge 0\). Hence
\[
D<-N,
\qquad
\text{so}
\qquad
-D>N.
\]

By \cref{lemma:supermodular}, we know that
\[
\frac{\partial^2 ER_{i}}{p_i p_j}\ge0, \ \forall \ j \ne i
\]
Hence, we have $N = \sum_{j \ne i} \frac{\partial^2 ER_{i}}{p_i p_j} \ge 0$. It follows that
\[
-D>N\ge0.
\]
Therefore,
\[
\sum_{j\ne i}
\frac{\partial BR_i}{\partial p_j}
=
-\frac{N}{D}
<1.
\]
\end{proof}

We now prove uniqueness. Since \(BR\) is increasing and maps the
price lattice into itself, Tarski's fixed point theorem implies that
the set of fixed points is a complete lattice. Let
\[
\bar{\bm p}
\quad\text{and}\quad
\underline{\bm p}
\]
be the largest and smallest fixed points, respectively.

Suppose, toward a contradiction, that
\(\bar{\bm p}\ne \underline{\bm p}\). Since
\(\bar{\bm p}\ge \underline{\bm p}\), there exists \(k\) such that
\[
\Delta_k
=
\bar p_k-\underline p_k
=
\max_i(\bar p_i-\underline p_i)
>0.
\]
Because both price vectors are fixed points,
\[
\Delta_k
=
BR_k(\bar p_{-k})-BR_k(\underline p_{-k}).
\]

Let
\[
p_{-k}(s)
=
\underline p_{-k}
+
s(\bar p_{-k}-\underline p_{-k}),
\qquad s\in[0,1].
\]
By the regularity of \(BR_k\) on the interior region, the map
\(s\mapsto BR_k(p_{-k}(s))\) is absolutely continuous, and hence
\[
BR_k(\bar p_{-k})-BR_k(\underline p_{-k})
=
\int_0^1
\sum_{j\ne k}
\frac{\partial BR_k}{\partial p_j}
\bigl(p_{-k}(s)\bigr)
(\bar p_j-\underline p_j)
\,ds,
\]
where the derivatives exist for a.e. \(s\).

By monotonicity of \(BR_k\),
\[
\frac{\partial BR_k}{\partial p_j}\bigl(p_{-k}(s)\bigr)\ge0
\quad\text{for a.e. }s.
\]
Moreover, by the choice of \(k\),
\[
\bar p_j-\underline p_j\le \Delta_k
\quad\text{for all }j.
\]
Therefore,
\[
\Delta_k
\le
\Delta_k
\int_0^1
\sum_{j\ne k}
\frac{\partial BR_k}{\partial p_j}
\bigl(p_{-k}(s)\bigr)
\,ds.
\]
By \cref{lemma:contraction},
\[
\sum_{j\ne k}
\frac{\partial BR_k}{\partial p_j}
\bigl(p_{-k}(s)\bigr)
<1
\quad\text{for a.e. }s.
\]
Hence
\[
\int_0^1
\sum_{j\ne k}
\frac{\partial BR_k}{\partial p_j}
\bigl(p_{-k}(s)\bigr)
\,ds
<1.
\]
Thus
\[
\Delta_k<\Delta_k,
\]
a contradiction. Therefore
\[
\bar{\bm p}=\underline{\bm p}.
\]
Since every equilibrium lies between the smallest and largest fixed
points, the equilibrium price vector is unique. Consequently, the
equilibrium outcome is unique.

\section{Omitted Proof of \cref{sec:n}}

\subsection{Proof of \cref{cor:non}}

Let \((p_1^\star,p_2^\star)\) be an equilibrium of the duopoly price-posting game. The existence of such an equilibrium follows from the proof of \cref{cor:eqm}. Let \(\bm M^\star\) be the posted-price menu profile corresponding to \((p_1^\star,p_2^\star)\). By construction, each firm is best responding to the other in the price-posting game. Moreover, by the posted-price best-response result in \cref{thm:br}, no firm can profitably deviate to a non-posted menu. Hence, firms are also best responding to each other under \(\bm M^\star\) in the duopoly non-exclusive menu-posting game.

\subsection{Proof of \cref{prop:demand}}

Suppose, toward a contradiction, that $\bm M = (M_1, M_2)$ is an equilibrium that is 
not demand-matching. Without loss of generality, we assume the violation 
occurs at firm 1: there exists an active contract $(q_1, t_1) \in M_1^A(\bm M)$ 
and a set of types $\bm\Theta_A' \subset \bm\Theta$ with strictly positive 
measure under $F$ such that for every $\bm\theta = (\theta_1, \theta_2) \in 
\bm\Theta_A'$, the agent optimally selects $(q_1, t_1)$ from $M_1$ paired with 
some contract $(q_2(q_1, \bm\theta), t_2(q_1, \bm\theta)) \in M_2$ satisfying 
$q_1 + q_2(q_1, \bm\theta) < 1$, where
\[
(q_2(q_1, \bm\theta), t_2(q_1, \bm\theta)) \in \argmax_{(q_2, t_2) \in M_2,\ 
q_2 < 1 - q_1} \big\{ \theta_2 q_2 - t_2 \big\}.
\]

Since $(q_1, t_1)$ is active, define $\underline{\theta}_1 = \inf\{\theta_1 
: \bm\theta \in \bm\Theta_A'\} + \eta > 0$, where $\eta>0$ is sufficiently small such that $\underline{\theta}_1 \in \bm \Theta_A'$. Let $\bm \Theta_A = \{ \bm \theta \subseteq \bm \Theta_A': \theta_1 \ge \underline{\theta}_1\}$. Consider firm 1 deviating to the augmented menu
\[
M_1' = M_1 \cup \{(q_1', t_1')\}, \quad q_1' = q_1 + \varepsilon, \quad t_1' 
= t_1 + \underline{\theta}_1 \varepsilon,
\]
 where $\varepsilon > 0$ is undetermined. 
 
 We know that there exists $\varepsilon' > 0$, such that $\{(q_2,t_2) \in M_2:q_2 \le 1-q_1\} = \{(q_2,t_2) \in M_2:q_2 \le 1-q_1 - \varepsilon'\}$ and $q_1+\varepsilon' + q_2(q_1+\varepsilon',\bm \theta) < 1$, for all $\bm \theta \in \bm \Theta_A$, as $M_1$ is a menu of finite contracts. Hence,  for all $\bm \theta \in \bm \Theta_A$, and $\varepsilon' \ge \varepsilon > 0$, we have
\begin{align*}
    q_1'\theta_1+q_2(q_1',\bm \theta)\theta_2 -t_1' -t_2(q_1',\bm\theta) &= q_1'\theta_1+q_2(q_1,\bm \theta)\theta_2 -t_1' -t_2(q_1,\bm\theta) \\
    &> q_1\theta_1+q_2(q_1,\bm \theta)\theta_2 -t_1 -t_2(q_1,\bm\theta)\\
    & \ge  q_1''\theta_1+q_2(q_1'',\bm \theta)\theta_2 -t_1'' -t_2(q_1'',\bm\theta), \ \forall (q''_1,t_1'') \in M_1
\end{align*}
The first equality comes from the fact that $\varepsilon' > 0$, such that $\{(q_2,t_2) \in M_2:q_2 \le 1-q_1\} = \{(q_2,t_2) \in M_2:q_2 \le 1-q_1 - \varepsilon\}$, so that the agent chooses the same contract from $M_2$, regardless of whether her choice from $M_1'$ is $q_1$ or $q_1'$. The rest of inequalities come directly from the definition. We can conclude that $(q_1',t_1')$ and $(q_2(q_1, 
\bm\theta), t_2(q_1, \bm\theta))$ are the optimal choice for all $\bm \theta \in \bm\Theta_A$ under $M_1'$, as $M_1' = M_1 \cup \{(q_1', t_1')\}$. As for all $\bm \theta \in \bm \Theta_A' \setminus \bm\Theta_A$, they stick to the original contract $(q_1,t_1)$.

Now, we consider any type $\bm\theta \notin \bm\Theta_A'$ that originally chooses the contract $(\hat{q}_1,\hat{t}_1) \ne (q_1,t_1)$ under $\bm M$, we must have, 
\begin{align*}
    \hat{q}_1 \theta_1 - \hat{t}_1 + q_2(\hat{q}_1,\bm \theta)\theta_2 - t_2(\hat{q}_1, \bm \theta) > q_1 \theta_1 - t_1 + q_2(q_1,\bm \theta)\theta_2 - t_2(q_1,\bm \theta)
\end{align*}
We know that there exists $\varepsilon' \ge \varepsilon''>0$ such that, 
\begin{align*}
    \hat{q}_1 \theta_1 - \hat{t}_1 + q_2(\hat{q}_1,\bm \theta)\theta_2 - t_2(\hat{q}_1, \bm \theta) > (q_1 +\varepsilon'') \theta_1 - t_1 + q_2(q_1,\bm \theta)\theta_2 - t_2(q_1,\bm \theta)
\end{align*}
Therefore, for all $\varepsilon'' \ge \varepsilon > 0$, we have
\begin{align*}
    \hat{q}_1 \theta_1 - \hat{t}_1 + q_2(\hat{q}_1,\bm \theta)\theta_2 - t_2(\hat{q}_1, \bm \theta) > q_1'  \theta_1 - t_1' + q_2(q_1',\bm \theta)\theta_2 - t_2(q_1',\bm \theta)
\end{align*}
This gives us the result that no type  $\bm\theta \notin \bm\Theta_A'$ will have an incentive to deviate from her original contracts under $M_1'$.

We conclude that there exists $\varepsilon > 0$ such that, for all $\bm \theta\in \bm \Theta_A$, they will deviate to $(q_1',t_1')$ and pay a higher transfer to firm 1, and for all $\bm \theta \not\in \bm \Theta_A$ will not change their original choice of contracts. Firm 1 can be strictly better off by deviating to $\bm M'$, which completes the proof that $\bm M$ is not an equilibrium.

\subsection{Proof of \cref{prop:finite}}

Let \((M_1,M_2)\) be an equilibrium where each firm offers a finite menu of contracts. By \cref{prop:demand}, every equilibrium is demand-matching, so every active contract pair satisfies \(q_1+q_2=1\). Since \(V(\Delta)\) is the upper envelope of finitely many linear functions, it is convex. Therefore, \(q_1(\Delta)=V'(\Delta)\) is non-decreasing almost everywhere, and the type space \([\underline{\Delta},\overline{\Delta}]\) is partitioned into finitely many intervals, on each of which the agent chooses a fixed contract pair.

Order the active contracts of firm \(1\) by quantity:
\[
0\le q_1<q_2<\cdots<q_k\le1,
\]
chosen by types on intervals
\[
[\Delta_0,\Delta_1],\ [\Delta_1,\Delta_2],\ \ldots,\ [\Delta_{k-1},\Delta_k],
\]
respectively, where \(\Delta_0=\underline{\Delta}\) and \(\Delta_k=\overline{\Delta}\). We want to show that the largest active contract must satisfy \(q_k=1\).

Suppose toward a contradiction that \(q_k<1\), so that the largest active contract of firm \(1\) is \((q_k,t_1^k)\), with \(q_k\in(0,1)\), chosen by types \(\Delta\in[\Delta_{k-1},\Delta_k]\). By demand-matching, firm \(2\)'s complementary active contract is \((1-q_k,t_2^k)\), chosen simultaneously by the same types.

Consider firm \(1\) adding a new contract \((1,p)\) to its menu. We set \(p\) such that type \(\hat{\Delta}\) is exactly indifferent between taking \((1,p)\) alone and choosing \((q_k,t_1^k)\) paired with \((1-q_k,t_2^k)\):
\[
\hat{\Delta}-p
=
\hat{\Delta}q_k-t_1^k-t_2^k.
\]
Thus,
\begin{equation}\label{eq:indiff_p_hat}
p
=
\hat{\Delta}(1-q_k)+t_1^k+t_2^k.
\end{equation}
Since \(\Delta_k=\overline{\Delta}>0\), \(q_k<1\), and \(t_2^k\ge0\), we have
\[
\Delta_k(1-q_k)+t_2^k>0.
\]
We can choose \(\hat{\Delta}\in(\Delta_{k-1},\Delta_k)\) sufficiently close to \(\Delta_k\) such that
\begin{align*}
    & \hat{\Delta}(1-q_k)+t_2^k>0\\
    \implies & p-t_1^k
=
\hat{\Delta}(1-q_k)+t_2^k
>
0.
\end{align*}

For any type \(\Delta\in[\Delta_{k-1},\Delta_k]\), the utility gain from \((1,p)\) over \((q_k,t_1^k)\) paired with \((1-q_k,t_2^k)\) is
\[
[\Delta-p]-[\Delta q_k-t_1^k-t_2^k]
=
(\Delta-\hat{\Delta})(1-q_k).
\]
Hence, all types \(\Delta>\hat{\Delta}\) strictly prefer \((1,p)\), while all types \(\Delta<\hat{\Delta}\) in the top interval do not prefer the new contract.

We now show that no type below \(\Delta_{k-1}\) wants to deviate to the new contract. For each contract \((q_j,t_1^j)\) with \(j\le k-1\), consider types \(\Delta\in[\Delta_{j-1},\Delta_j]\). The utility gain from \((1,p)\) over \((q_j,t_1^j)\) paired with \((1-q_j,t_2^j)\) is
\[
[\Delta-p]-[\Delta q_j-t_1^j-t_2^j]
=
\Delta(1-q_j)-p+t_1^j+t_2^j.
\]
This expression is increasing in \(\Delta\), since \(1-q_j\ge0\). Therefore, it is enough to evaluate it at the upper boundary \(\Delta=\Delta_j\le\Delta_{k-1}\).

By the chain of indifference conditions, at each threshold \(\Delta_l\), for \(l=j,j+1,\ldots,k-1\), the agent is indifferent between consecutive contract pairs:
\[
\Delta_l q_l-t_1^l-t_2^l
=
\Delta_l q_{l+1}-t_1^{l+1}-t_2^{l+1}.
\]
For each such \(l\), this implies
\[
(\Delta_j q_{l+1}-t_1^{l+1}-t_2^{l+1})
-
(\Delta_j q_l-t_1^l-t_2^l)
=
(\Delta_j-\Delta_l)(q_{l+1}-q_l).
\]
Summing these differences from \(l=j\) to \(k-1\) yields
\[
\Delta_j q_j-t_1^j-t_2^j
=
\Delta_j q_k-t_1^k-t_2^k
-
\sum_{l=j}^{k-1}
(\Delta_j-\Delta_l)(q_{l+1}-q_l).
\]

Therefore, the utility difference at \(\Delta=\Delta_j\) is
\begin{align*}
&[\Delta_j-p]-[\Delta_j q_j-t_1^j-t_2^j] \\
=&
\Delta_j-p
-
\left[
\Delta_j q_k-t_1^k-t_2^k
-
\sum_{l=j}^{k-1}
(\Delta_j-\Delta_l)(q_{l+1}-q_l)
\right] \\
=&
\Delta_j-p-\Delta_j q_k+t_1^k+t_2^k
+
\sum_{l=j}^{k-1}
(\Delta_j-\Delta_l)(q_{l+1}-q_l) \\
=&
(\Delta_j-\hat{\Delta})(1-q_k)
+
\sum_{l=j}^{k-1}
(\Delta_j-\Delta_l)(q_{l+1}-q_l).
\end{align*}
The last equality uses the definition of \(p\) in \eqref{eq:indiff_p_hat}. Since \(\Delta_j\le\Delta_{k-1}<\hat{\Delta}\), the first term is strictly negative. Moreover, for all \(l\ge j\), we have \(\Delta_j\le\Delta_l\) and \(q_{l+1}>q_l\), so each term in the sum is non-positive. Hence
\[
[\Delta_j-p]-[\Delta_j q_j-t_1^j-t_2^j]<0.
\]
Since the utility gain from the new contract is increasing in \(\Delta\) on the interval \([\Delta_{j-1},\Delta_j]\), it is strictly negative throughout that interval. Therefore, no type below \(\Delta_{k-1}\) wants to deviate to the new contract.

Under the deviation, exactly the types \(\Delta\in(\hat{\Delta},\Delta_k]\) switch from \((q_k,t_1^k)\) to \((1,p)\). The change in firm \(1\)'s expected revenue is
\[
\Delta\Pi_1
=
(p-t_1^k)\bigl[H(\Delta_k)-H(\hat{\Delta})\bigr].
\]
Since \(p-t_1^k>0\) and \(\hat{\Delta}<\Delta_k\), we have
\[
\Delta\Pi_1>0.
\]
Thus firm \(1\) has a profitable deviation, contradicting the assumption that \((M_1,M_2)\) is an equilibrium. Therefore, the largest active contract of firm \(1\) must satisfy \(q_k=1\).

\subsection{Proof of \cref{prop:finite_unique}}

By \cref{prop:finite}, both firms' active menus contain posted-price contracts $(1, p_1)$ and $(1, p_2)$. Suppose for contradiction that the active menus contain interior contracts. Order firm 1's active contracts by allocation:
\[
0 = q_0 < q_1 < q_2 < \cdots < q_{K-1} < q_K = 1,
\]
with corresponding transfers $t_1^0 = 0,\, t_1^1, \ldots, t_1^{K-1},\, t_1^K = p_1$. We assume $K \geq 2$, so that at least one interior contract exists. By demand-matching (\cref{prop:demand}), firm 2's complementary active contracts have allocations $1 - q_k$ with transfers $t_2^k$, and in particular $t_2^0 = p_2$ and $t_2^K = 0$.

Since $V(\Delta)$ is the upper envelope of finitely many affine functions, it is piecewise linear and convex, so the allocation $q_1(\Delta) = V'(\Delta)$ almost everywhere is a non-decreasing step function. The type space is partitioned as follows:
\begin{itemize}
\item Types with $\Delta < \Delta_0$ buy exclusively from firm 2, choosing $(0, 0)$ from firm 1 and $(1, p_2)$ from firm 2.
\item Types with $\Delta \in [\Delta_{k-1}, \Delta_k]$, for $k = 1, \ldots, K-1$, choose interior contract $k$ from firm 1, namely $(q_k, t_1^k)$, paired with $(1 - q_k, t_2^k)$ from firm 2.
\item Types with $\Delta > \Delta_{K-1}$ buy exclusively from firm 1, choosing $(1, p_1)$ from firm 1 and $(0,0)$ from firm 2.
\end{itemize}

We denote the step sizes by $\alpha_k = q_k - q_{k-1} > 0$ for $k = 1, \ldots, K$, and the total transfer for contract $k$ by $s_k = t_1^k + t_2^k$. We know that, at each boundary $\Delta_k$ for $k = 0, 1, \ldots, K-1$, the agent is indifferent between adjacent contracts $k$ and $k+1$:
\[
\Delta_k \, q_k - s_k = \Delta_k \, q_{k+1} - s_{k+1}.
\]
This yields:
\begin{equation}\label{eq:indiff}
s_{k+1} - s_k = \Delta_k \, \alpha_{k+1}, \quad k = 0, 1, \ldots, K-1.
\end{equation}

Give the menu profile we can write firm 1's expected revenue as
\[
\Pi_1 = \sum_{j=1}^{K-1} t_1^j \big[H(\Delta_j) - H(\Delta_{j-1})\big] + p_1 \big[1 - H(\Delta_{K-1})\big].
\]
Changing $t_1^k$ shifts the boundaries $\Delta_{k-1}$ and $\Delta_k$ via the indifference conditions. From \eqref{eq:indiff}:
\[
\frac{\partial \Delta_{k-1}}{\partial t_1^k} = \frac{1}{\alpha_k}, \qquad \frac{\partial \Delta_k}{\partial t_1^k} = \frac{-1}{\alpha_{k+1}}.
\]
Since the given menu profile is an equilibrium, firm 1 must maximize its expected revenue. By first order condition, we must have $\frac{\partial \Pi_1}{\partial t_1^k} = 0$, which give us
\begin{equation}\label{eq:F1k}
H(\Delta_k) - H(\Delta_{k-1}) + \frac{h(\Delta_{k-1})}{\alpha_k}(t_1^{k-1} - t_1^k) + \frac{h(\Delta_k)}{\alpha_{k+1}}(t_1^{k+1} - t_1^k) = 0. \tag{F1-$k$}
\end{equation}

Similarly, we can write firm 2's expected revenue as
\[
\Pi_2 = p_2 \, H(\Delta_0) + \sum_{j=1}^{K-1} t_2^j \big[H(\Delta_j) - H(\Delta_{j-1})\big].
\]
By the same calculation, with $\frac{\partial \Delta_{k-1}}{\partial t_2^k} = \frac{1}{\alpha_k}$ and $\frac{\partial \Delta_k}{\partial t_2^k} = \frac{-1}{\alpha_{k+1}}$, setting $\frac{\partial \Pi_2}{\partial t_2^k} = 0$ yields:
\begin{equation}\label{eq:F2k}
H(\Delta_k) - H(\Delta_{k-1}) + \frac{h(\Delta_{k-1})}{\alpha_k}(t_2^{k-1} - t_2^k) + \frac{h(\Delta_k)}{\alpha_{k+1}}(t_2^{k+1} - t_2^k) = 0. \tag{F2-$k$}
\end{equation}

Note that, since the given profile is an equilibrium, \eqref{eq:F1k} and \eqref{eq:F2k} must hold at the same time. Adding \eqref{eq:F1k} and \eqref{eq:F2k}, and writing $s_k = t_1^k + t_2^k$:
\begin{equation}\label{eq:combined}
2\big[H(\Delta_k) - H(\Delta_{k-1})\big] + \frac{h(\Delta_{k-1})}{\alpha_k}(s_{k-1} - s_k) + \frac{h(\Delta_k)}{\alpha_{k+1}}(s_{k+1} - s_k) = 0.
\end{equation}

Substituting the indifference conditions \eqref{eq:indiff}, which give $s_{k-1} - s_k = -\Delta_{k-1}\alpha_k$ and $s_{k+1} - s_k = \Delta_k \alpha_{k+1}$:
\begin{align*}
    & 2\big[H(\Delta_k) - H(\Delta_{k-1})\big] - h(\Delta_{k-1})\,\Delta_{k-1} + h(\Delta_k)\,\Delta_k = 0 \\
    \implies & \Delta_{k-1}h(\Delta_{k-1}) + 2 H(\Delta_{k-1}) = \Delta_kh(\Delta_{k}) + 2 H(\Delta_{k})
\end{align*}

The above condition must hold for every $k = 1, \ldots, K-1$. Defining $\varphi(\Delta) =  \Delta h(\Delta) + 2 H(\Delta)$, if $\varphi(\Delta)$ is strictly monotone, then there must exist a contradiction.

\subsection{Proof of \cref{prop:continuous}}

Let \((t_1,t_2)\) be an equilibrium. We know that \(t_1\) must maximize firm 1's expected profit, taking \(t_2\) as given. We first write firm 1's expected profit:
\[
\int_{\underline{\Delta}}^{\overline{\Delta}}
\Big[ t_1(q_1(\Delta)) - c_1 q_1(\Delta) \Big] h(\Delta)\, d\Delta.
\]

We also know that
\[
t_1(q_1(\Delta))
=
\Delta q_1(\Delta)-t_2(1-q_1(\Delta))-V(\Delta).
\]

Substituting this into firm 1's expected profit gives
\[
\int_{\underline{\Delta}}^{\overline{\Delta}}
\Big[
\Delta q_1(\Delta)-t_2(1-q_1(\Delta))-V(\Delta)-c_1q_1(\Delta)
\Big]h(\Delta)\,d\Delta.
\]

We use integration by parts on the term \(-\int V(\Delta)h(\Delta)d\Delta\). Let \(u=V(\Delta)\), so \(du=q_1(\Delta)d\Delta\), and let \(dv=h(\Delta)d\Delta\), so \(v=-(1-H(\Delta))\). We have
\[
-\int_{\underline{\Delta}}^{\overline{\Delta}} V(\Delta)h(\Delta)d\Delta
=
-V(\underline{\Delta})
-
\int_{\underline{\Delta}}^{\overline{\Delta}}
q_1(\Delta)\frac{1-H(\Delta)}{h(\Delta)}h(\Delta)d\Delta.
\]

Substituting this back into firm 1's expected profit, we have
\[
\int_{\underline{\Delta}}^{\overline{\Delta}}
\left[
\Delta q_1(\Delta)-t_2(1-q_1(\Delta))-c_1q_1(\Delta)
-q_1(\Delta)\frac{1-H(\Delta)}{h(\Delta)}
\right]h(\Delta)d\Delta
-
V(\underline{\Delta}).
\]

Firm 1 performs pointwise maximization with respect to \(q_1\). Taking the first-order condition inside the bracket gives
\[
\Delta+t_2'(1-q_1(\Delta))-c_1-\frac{1-H(\Delta)}{h(\Delta)}=0.
\]

We can derive firm 2's best-response condition in a similar fashion. We first write firm 2's expected profit:
\[
\int_{\underline{\Delta}}^{\overline{\Delta}}
\Big[
t_2(1-q_1(\Delta))-c_2(1-q_1(\Delta))
\Big]h(\Delta)d\Delta.
\]

Knowing that
\[
t_2(1-q_1(\Delta))
=
\Delta q_1(\Delta)-t_1(q_1(\Delta))-V(\Delta),
\]
and substituting this back into firm 2's expected profit gives
\[
\int_{\underline{\Delta}}^{\overline{\Delta}}
\left[
\Delta q_1(\Delta)-t_1(q_1(\Delta))-V(\Delta)-c_2(1-q_1(\Delta))
\right]h(\Delta)\,d\Delta.
\]

We need to handle the term \(\int V(\Delta)h(\Delta)\,d\Delta\). Using integration by parts yields
\[
\int_{\underline{\Delta}}^{\overline{\Delta}}V(\Delta)h(\Delta)\,d\Delta
=
V(\overline{\Delta})
-
\int_{\underline{\Delta}}^{\overline{\Delta}}H(\Delta)q_1(\Delta)\,d\Delta.
\]

Substituting this result back into firm 2's profit expression, we can pull out \(V(\overline{\Delta})\), since it is a constant, and group the integral terms:
\[
-V(\overline{\Delta})
+
\int_{\underline{\Delta}}^{\overline{\Delta}}
\left[
\Delta q_1(\Delta)-t_1(q_1(\Delta))-c_2(1-q_1(\Delta))
+\frac{H(\Delta)}{h(\Delta)}q_1(\Delta)
\right]h(\Delta)\,d\Delta.
\]

To maximize this integral, firm 2 maximizes the expression inside the brackets pointwise with respect to \(q_1\). Taking the derivative with respect to \(q_1\) and setting it equal to zero gives
\[
\Delta-t_1'(q_1(\Delta))+c_2+\frac{H(\Delta)}{h(\Delta)}=0.
\]

Now suppose that type \(\Delta\) chooses an interior solution \(q_1^\star(\Delta)\in(0,1)\). The agent's first-order condition requires
\[
\Delta=t_1'(q_1(\Delta))-t_2'(1-q_1(\Delta)).
\]

Since \(V(\Delta)\) is convex and \(V'(\Delta)=q_1(\Delta)\), differentiating the agent's first-order condition with respect to \(\Delta\) gives
\[
1=
\left[
t_1''(q_1(\Delta))+t_2''(1-q_1(\Delta))
\right]q_1'(\Delta).
\]

Now we differentiate the best-response conditions with respect to \(\Delta\), which gives
\begin{align*}
1-t_2''(1-q_1(\Delta))q_1'(\Delta)
&=
\frac{d}{d\Delta}\left[\frac{1-H(\Delta)}{h(\Delta)}\right],\\
1-t_1''(q_1(\Delta))q_1'(\Delta)
&=
-\frac{d}{d\Delta}\left[\frac{H(\Delta)}{h(\Delta)}\right].
\end{align*}

Adding the above two equations gives
\[
\left[
t_1''(q_1(\Delta))+t_2''(1-q_1(\Delta))
\right]q_1'(\Delta)
=
2+\frac{d}{d\Delta}\left[\frac{H(\Delta)}{h(\Delta)}\right]
-
\frac{d}{d\Delta}\left[\frac{1-H(\Delta)}{h(\Delta)}\right].
\]

Substituting the differentiated agent's first-order condition into the left-hand side, we have
\begin{align*}
&1
=
2+\frac{d}{d\Delta}\left[\frac{H(\Delta)}{h(\Delta)}\right]
-
\frac{d}{d\Delta}\left[\frac{1-H(\Delta)}{h(\Delta)}\right] \\
\implies &
\frac{d}{d\Delta}\left[\frac{1-H(\Delta)}{h(\Delta)}\right]
-
\frac{d}{d\Delta}\left[\frac{H(\Delta)}{h(\Delta)}\right]
=1 \\
\implies &
\frac{d}{d\Delta}\left[\frac{1-2H(\Delta)}{h(\Delta)}\right]
=1.
\end{align*}

That is, \(\psi'(\Delta)=1\). This is a necessary condition for any type \(\Delta\) to choose an interior solution. Therefore, the set of types choosing interior solutions is contained in
\[
\mathcal I
=
\left\{
\Delta\in[\underline{\Delta},\overline{\Delta}]
:
\psi'(\Delta)=1
\right\}.
\]

If \(\psi'(\Delta) \ne 1\) almost everywhere, then \(\mathcal I\) is a measure-zero set under \(h\), and therefore agents choose interior solutions with probability zero.
\section{Omitted Proof of \cref{sec:cor}}

\subsection{Proof of \cref{prop:correlated}}

The proof follows the structure of \cref{thm:br}, with the modification needed to handle the direct dependence of the winning probability on \(\theta_i\). Firm \(i\)'s expected revenue can be written as
\[
ER_i
=
\int_{0}^{\bar{\theta}_i}
\big[\theta_i q_i(\theta_i)-U_i(\theta_i)\big]
P_i(\theta_i,U_i(\theta_i))f_i(\theta_i)\,d\theta_i,
\]
where
\[
P_i(\theta_i,U_i)
=
\Pr\left(
\max_{j\ne i}U_j(\theta_j)\le U_i
\,\middle|\,
\theta_i
\right)
\]
is the conditional probability that firm \(i\) wins, given own type \(\theta_i\) and offered utility \(U_i\). Unlike in the independent case, \(P_i\) depends on \(\theta_i\) both directly and through \(U_i(\theta_i)\).

Setting up the Hamiltonian as before,
\[
\mathcal H(\theta_i)
=
q_i(\theta_i)
\big[
\theta_iP_i(\theta_i,U_i)f_i(\theta_i)
+
\lambda(\theta_i)
\big]
-
U_i(\theta_i)P_i(\theta_i,U_i)f_i(\theta_i),
\]
the switching function is
\[
\Phi(\theta_i)
=
\theta_iP_i(\theta_i,U_i(\theta_i))f_i(\theta_i)
+
\lambda(\theta_i),
\]
and the costate evolves according to
\[
\lambda'(\theta_i)
=
P_i(\theta_i,U_i)f_i(\theta_i)
-
\big[\theta_i q_i(\theta_i)-U_i(\theta_i)\big]
\frac{\partial P_i}{\partial U_i}(\theta_i,U_i)f_i(\theta_i),
\]
with transversality condition
\[
\lambda(\bar{\theta}_i)=0.
\]

Differentiating \(\Phi(\theta_i)\) with respect to \(\theta_i\), using \(U_i'(\theta_i)=q_i(\theta_i)\), and substituting the costate equation yields
\begin{equation}\label{eq:phi-prime-correlated}
\Phi'(\theta_i)
=
\big[2f_i(\theta_i)+\theta_i f_i'(\theta_i)\big]P_i
+
\theta_i\frac{\partial P_i}{\partial\theta_i}f_i(\theta_i)
+
U_i(\theta_i)\frac{\partial P_i}{\partial U_i}f_i(\theta_i).
\end{equation}
Compared to the independent case, \eqref{eq:phi-prime-correlated} contains the additional term
\[
\theta_i\frac{\partial P_i}{\partial\theta_i}f_i(\theta_i).
\]
We now show that this term is non-negative under stochastic decreasingness in own type.

For a given \(U_i\), define
\[
A(U_i)
=
\left\{
\bm\theta_{-i}:
\max_{j\ne i}U_j(\theta_j)\le U_i
\right\}.
\]
Since each \(U_j\) is non-decreasing, this set can be written as a lower rectangle using generalized inverses:
\[
A(U_i)
=
\prod_{j\ne i}
\left[0,U_j^{-1}(U_i)\right].
\]
Therefore,
\[
P_i(\theta_i,U_i)
=
F_{-i\mid i}
\left(
(U_j^{-1}(U_i))_{j\ne i}
\,\middle|\,
\theta_i
\right).
\]
By \cref{def:sdcd}, the conditional CDF \(F_{-i\mid i}(\cdot\mid\theta_i)\) is non-decreasing in \(\theta_i\). Hence
\[
\frac{\partial P_i}{\partial\theta_i}\ge0
\]
wherever the derivative exists.

Density-regularity gives
\[
2f_i(\theta_i)+\theta_i f_i'(\theta_i)>0.
\]
Moreover, \(P_i\ge0\), \(\theta_i\ge0\), \(f_i>0\), \(U_i\ge0\), and \(\partial P_i/\partial U_i\ge0\). Together with \(\partial P_i/\partial\theta_i\ge0\), all terms in \eqref{eq:phi-prime-correlated} are non-negative. Therefore,
\[
\Phi'(\theta_i)\ge0.
\]
The same no-singular-arc argument as in \cref{lemma:no_singular_arc} then implies that the optimal allocation rule is a step function, up to changes on a measure-zero set. The rest of the proof follows from the proof of \cref{thm:br}.

\subsection{Proof of \cref{prop:cor_post}}

Fix a posted-price profile \(p_{-i}\). As shown above, firm \(i\)'s winning probability is
\[
P_i^p(\theta_i,U)
=
F_{-i|i}(U+p_{-i}\mid\theta_i).
\]
We can set up the Hamiltonian as in the proof of \cref{prop:correlated}. The switching function associated with firm \(i\)'s optimal control problem is
\[
\Phi_i(\theta_i)
=
\theta_i f_i(\theta_i)P_i^p(\theta_i,U_i(\theta_i))
+
\lambda_i(\theta_i),
\]
and its derivative is
\[
\begin{aligned}
\Phi_i'(\theta_i)
&=
\left[
2f_i(\theta_i)+\theta_i f_i'(\theta_i)
\right]
P_i^p(\theta_i,U_i(\theta_i)) \\
&\quad
+
\theta_i f_i(\theta_i)
\frac{\partial P_i^p}{\partial \theta_i}
(\theta_i,U_i(\theta_i)) \\
&\quad
+
U_i(\theta_i)f_i(\theta_i)
\frac{\partial P_i^p}{\partial U}
(\theta_i,U_i(\theta_i)).
\end{aligned}
\]
Because \(P_i^p(\theta_i,U)=F_{-i|i}(U+p_{-i}\mid\theta_i)\), the first two terms can be written as
\[
\left[
2f_i(\theta_i)+\theta_i f_i'(\theta_i)
\right]
F_{-i|i}(U_i(\theta_i)+p_{-i}\mid\theta_i)
+
\theta_i f_i(\theta_i)
\frac{\partial}{\partial \theta_i}
F_{-i|i}(U_i(\theta_i)+p_{-i}\mid\theta_i).
\]
By correlation-adjusted density-regularity, this expression is non-negative for any trajectory \((q_i,U_i,\lambda_i)\). The third term is also non-negative because \(U_i\ge0\), \(f_i>0\), and \(P_i^p\) is non-decreasing in \(U\). Hence
\[
\Phi_i'(\theta_i)\ge0.
\]
Thus, the switching function \(\Phi_i\) is non-decreasing. The rest of the proof follows from the proof of \cref{thm:br}.

\subsection{Proof of \cref{cor:eqm_cor}}
The proof follows from \cref{cor:eqm}. As we assume that the joint density is continuous and strictly positive, we can show that the best response is a continuous function. Then, Brouwer’s Fixed Point Theorem applies.

\subsection{Proof of \cref{prop:per_cor}}

By our setting, there exists a scalar random variable \(\theta\) and constants
\(c_1,\ldots,c_n\) such that
\[
    \theta_i=\theta+c_i,\qquad i=1,\ldots,n.
\]
For every opponent \(j\neq i\), we have
\[
    \theta_j=\theta_i+c_j-c_i.
\]
Suppose that all opponents post prices \(p_j\), \(j\neq i\). Then opponent \(j\) gives the agent utility
\[
    U_j(\theta_j)
    =
    (\theta_j-p_j)_+
    =
    (\theta_i+c_j-c_i-p_j)_+.
\]
Hence the best utility available from the opponents is
\[
    O_i(\theta_i)
    =
    \max_{j\neq i}(\theta_i+c_j-c_i-p_j)_+
    =
    (\theta_i-\bar p_{-i})_+,
\]
where
\[
    \bar p_{-i}
    =
    \min_{j\neq i}\{p_j-c_j+c_i\}.
\]
By the maintained assumption that \(p_j>c_j-c_i\) for all \(j\neq i\), we have
\[
    \bar p_{-i}>0.
\]

Firm \(i\) wins type \(\theta_i\) if and only if
\[
    U_i(\theta_i)>O_i(\theta_i).
\]
Thus define the winning indicator
\[
    P_i(\theta_i,U_i)
    =
    \mathbf 1\{U_i>O_i(\theta_i)\}.
\]
Firm \(i\)'s expected revenue can be written as
\[
ER_i
=
\int_{\Theta_i}
\left[
\theta_iq_i(\theta_i)-U_i(\theta_i)
\right]
P_i(\theta_i,U_i(\theta_i))
f_i(\theta_i)\,d\theta_i,
\]
subject to the state equation
\[
U_i'(\theta_i)=q_i(\theta_i),
\qquad
U_i(0)=0.
\]

The Hamiltonian is
\[
\mathcal H(\theta_i,U_i,q_i,\lambda_i)
=
\left[
\theta_iq_i-U_i
\right]
P_i(\theta_i,U_i)f_i(\theta_i)
+
\lambda_iq_i.
\]
Equivalently,
\[
\mathcal H
=
q_i
\left[
\theta_iP_i(\theta_i,U_i)f_i(\theta_i)
+
\lambda_i
\right]
-
U_iP_i(\theta_i,U_i)f_i(\theta_i).
\]
Hence the switching function is
\[
\Phi_i(\theta_i)
=
\theta_iP_i(\theta_i,U_i(\theta_i))f_i(\theta_i)
+
\lambda_i(\theta_i).
\]
The Hamiltonian maximization condition implies
\[
\Phi_i(\theta_i)>0
\implies
q_i^\star(\theta_i)=1,
\qquad
\Phi_i(\theta_i)<0
\implies
q_i^\star(\theta_i)=0.
\]

Since \(P_i\) is an indicator, it is not differentiable at contact points where
\[
U_i(\theta_i)=O_i(\theta_i).
\]
Away from such contact points, \(P_i\) is locally constant. Therefore, the costate equation reduces to
\[
\lambda_i'(\theta_i)
=
P_i(\theta_i,U_i(\theta_i))f_i(\theta_i),
\]
with transversality condition
\[
\lambda_i(\bar\theta_i)=0.
\]
Thus
\[
\lambda_i(\theta_i)
=
-\int_{\theta_i}^{\bar\theta_i}
P_i(s,U_i(s))f_i(s)\,ds.
\]

In particular, if firm \(i\) loses type \(\theta_i\), then
\[
P_i(\theta_i,U_i(\theta_i))=0,
\]
and therefore
\[
\Phi_i(\theta_i)
=
\lambda_i(\theta_i)
=
-\int_{\theta_i}^{\bar\theta_i}
P_i(s,U_i(s))f_i(s)\,ds
\le0.
\]
Hence
\[
\Phi_i(\theta_i)>0
\implies
P_i(\theta_i,U_i(\theta_i))=1.
\]
That is, whenever the switching function is positive, firm \(i\) strictly wins that type.

On any strict winning interval, we have \(P_i=1\). Thus
\[
\Phi_i(\theta_i)
=
\theta_i f_i(\theta_i)+\lambda_i(\theta_i).
\]
Since \(\lambda_i'(\theta_i)=f_i(\theta_i)\) on such an interval, differentiating gives
\[
\Phi_i'(\theta_i)
=
2f_i(\theta_i)+\theta_i f_i'(\theta_i).
\]
By density-regularity,
\[
2f_i(\theta_i)+\theta_i f_i'(\theta_i)>0.
\]
Therefore,
\[
\Phi_i'(\theta_i)>0
\]
on every strict winning interval.

We next show that the set on which \(\Phi_i=0\) cannot contain a non-degenerate interval. Suppose, toward a contradiction, that there exists an interval \([a,b]\subseteq\Theta_i\), with \(b>a\), such that
\[
\Phi_i(\theta_i)=0
\quad
\text{for all } \theta_i\in[a,b].
\]
This interval cannot intersect a strict winning interval on a set of positive measure, since \(\Phi_i'>0\) on every strict winning interval. Hence \(P_i=0\) almost everywhere on \([a,b]\). On this interval,
\[
\Phi_i(\theta_i)
=
-\int_{\theta_i}^{\bar\theta_i}
P_i(s,U_i(s))f_i(s)\,ds.
\]
Since \(\Phi_i(\theta_i)=0\) for all \(\theta_i\in[a,b]\), it follows that
\[
\int_{\theta_i}^{\bar\theta_i}
P_i(s,U_i(s))f_i(s)\,ds
=
0
\quad
\text{for all } \theta_i\in[a,b].
\]
Therefore \(P_i(s,U_i(s))=0\) almost everywhere on \([\theta_i,\bar\theta_i]\) for every \(\theta_i\in[a,b]\). In particular, firm \(i\) wins only a measure-zero set of types above \(a\). Since \(\Phi_i\le0\) on losing intervals and \(\Phi_i=0\) on \([a,b]\), the Hamiltonian maximization condition implies that the allocation is zero almost everywhere on and below this region. Hence this trajectory yields zero expected revenue.

This is suboptimal. Since \(\bar p_{-i}>0\), firm \(i\) can post any price \(p_i\in(0,\bar p_{-i})\). Then all sufficiently high types strictly prefer firm \(i\)'s posted-price contract to the opponents' outside option, and firm \(i\) obtains strictly positive expected revenue. This contradiction shows that
\[
\{\theta_i:\Phi_i(\theta_i)=0\}
\]
has measure zero.

It remains to show that the set on which \(\Phi_i>0\) is an upper interval. Suppose \(\Phi_i(\theta_0)>0\). Then, as shown above,
\[
P_i(\theta_0,U_i(\theta_0))=1,
\]
so firm \(i\) strictly wins type \(\theta_0\):
\[
U_i(\theta_0)>O_i(\theta_0).
\]
Moreover, the Hamiltonian maximization condition implies that whenever \(\Phi_i>0\),
\[
q_i^\star(\theta_i)=1.
\]
Hence, on any interval on which \(\Phi_i>0\),
\[
U_i'(\theta_i)=1.
\]
The outside option is
\[
O_i(\theta_i)=(\theta_i-\bar p_{-i})_+,
\]
so \(O_i'(\theta_i)\le1\) wherever \(O_i\) is differentiable. Therefore, on any interval on which \(\Phi_i>0\),
\[
\frac{d}{d\theta_i}
\left[
U_i(\theta_i)-O_i(\theta_i)
\right]
=
U_i'(\theta_i)-O_i'(\theta_i)
=
1-O_i'(\theta_i)
\ge0.
\]
Thus, after \(\Phi_i\) becomes positive, the gap \(U_i(\theta_i)-O_i(\theta_i)\) cannot cross downward from positive to non-positive. Hence firm \(i\) cannot move from a strict winning region to a losing region after \(\Phi_i\) has become positive.

Since \(\Phi_i\le0\) on losing intervals and \(\Phi_i\) is strictly increasing on strict winning intervals, it follows that once \(\Phi_i\) becomes positive, it remains positive for all higher types. Therefore,
\[
\{\theta_i:\Phi_i(\theta_i)>0\}
\]
is an upper interval.

Since the zero set of \(\Phi_i\) has measure zero, the optimal allocation rule is pinned down almost everywhere by the sign of \(\Phi_i\). Hence there exists a cutoff \(\theta_i^\star\in\Theta_i\) such that
\[
q_i^\star(\theta_i)
=
\begin{cases}
0, & \theta_i<\theta_i^\star,\\
1, & \theta_i\ge\theta_i^\star,
\end{cases}
\quad
\text{a.e.}
\]
Since \(U_i(0)=0\), the envelope condition gives
\[
U_i^\star(\theta_i)
=
\int_0^{\theta_i}q_i^\star(s)\,ds
=
\max\{\theta_i-\theta_i^\star,0\}.
\]
The corresponding transfer is
\[
t_i^\star(\theta_i)
=
\theta_iq_i^\star(\theta_i)-U_i^\star(\theta_i)
=
\begin{cases}
0, & \theta_i<\theta_i^\star,\\
\theta_i^\star, & \theta_i\ge\theta_i^\star.
\end{cases}
\]
Therefore, every best response induces the same outcome as the posted-price menu
\[
M_i^\star=\{(0,0),(1,\theta_i^\star)\}.
\]
This completes the proof.

\end{document}